\documentclass[11pt,a4paper]{amsart}%
\usepackage{amsmath, amsfonts, amsthm,color, latexsym, amssymb,geometry}
\usepackage{amscd}
\usepackage{amsmath}
\usepackage{amsfonts}
\usepackage{amssymb}
\usepackage{graphicx}%

\usepackage[initials]{amsrefs}

\providecommand{\U}[1]{\protect\rule{.1in}{.1in}}
\providecommand{\U}[1]{\protect\rule{.1in}{.1in}}
\providecommand{\U}[1]{\protect\rule{.1in}{.1in}}
\newtheorem{theorem}{Theorem}[section]
\newtheorem{proposition}[theorem]{Proposition}
\newtheorem{corollary}[theorem]{Corollary}

\newtheorem{remark}[theorem]{Remark}

\newtheorem{lemma}[theorem]{Lemma}
\newtheorem{definition}[theorem]{Definition}
\usepackage{units}
\providecommand{\keywords}[1]{\textbf{\textit{Index terms---}} #1}
\begin{document}

\title{Separability for Weak Irreducible matrices}
\author{Daniel Cariello}


\date{}
\maketitle

\begin{abstract}
This paper is devoted to the study of the separability problem in the field of Quantum information theory.  We deal mainly with the bipartite finite dimensional case and with two types of matrices, one of them being the PPT matrices (see definitions \ref{definitionSPC} and \ref{definitionPPT}). 
We proved that many results holds for both types. If these matrices have specific Hermitian Schmidt decompositions then  the matrices are separable in a very strong sense (see theorem \ref{theomiracle1} and corollary \ref{corollarytheomiracle1}). We proved that both types have what we call \textbf{split decompositions} (see theorem \ref{splitdecompositionSPC} and \ref{splitdecompositionPPT}). 

We defined the notion of weak irreducible matrix (see definition \ref{definitionweakirreducible}), based on the concept of irreducible state defined recently in \cite{chen1} and \cite{chen2}.
These split decomposition theorems together with the notion of weak irreducible matrix, imply that these matrices are weak irreducible or a sum of weak irreducible matrices of the same type (see theorem \ref{theoremsumofweakirreducible}). The separability problem for these types of matrices can be reduced to the set of weak irreducible matrices of the same type (see proposition \ref{reductionfinal}).
We also provided a complete description of weak irreducible matrices of both types (see theorem \ref{conditionsirreducible2}).

Using the fact that every positive semidefinite Hermitian matrix with tensor rank 2 is separable (see theorem \ref{theoreminimalseparabilitytensorrank2}), we found sharp inequalites providing separability for both types (see theorems \ref{inequality1} and \ref{inequality2}).

\end{abstract}

\tableofcontents

\section*{Introduction}
The Separability Problem is a well established problem in the field of Quantum Information Theory by its importance and difficulty. Its aim is to find a criterion to distinguish the separable states from the entangled states (see definition \ref{definitionseparability}). In this paper we deal only with the bipartite finite dimensional case, therefore the states are elements in $M_{k}\otimes M_{m}$, which can be interpreted as matrices in $M_{ km}$ via the Kronecker product (see notation at the end of this introduction).

This problem was completely solved by Horodecki in the space $M_{2}\otimes M_{m}$ for $m=2$ or $3$, by the so-called PPT criterion (see \cite{horodeckifamily}). This criterion states that a matrix $ A=\sum_{i=1}^k A_i\otimes D_i\in M_{2}\otimes M_{m}$ (for $m=2$ or $3$) is separable if and only if $A$ and $A^{t_1}$ are positive semidefinite Hermitian matrices, where $A^{t_1}=\sum_{i=1}^k A_i^t\otimes D_i$ is the partial transposition of $A$. Such matrices are called PPT matrices (see definition \ref{definitionPPT}).

The general case, even the finite dimensional case, still is a great challenge. Algorithms were developed to solve the separability problem, but it is known that this problem is NP-hard (see \cite{gurvits}).

Therefore any restriction of the problem to a smaller set of matrices is, certainly, important. For example, Peres in \cite{peres} was the first to notice the importance of the PPT property which was later proved to be necessary and sufficient for separability in $M_{2}\otimes M_{m}$ for $m=2$ or $3$, in \cite{horodeckifamily}.

Another remarkable reduction was obtained in \cite{leinaas} for the positive definite case  in $M_{k}\otimes M_{m}$. The authors proved that  to find the separable positive definite Hermitian matrices we only need to distinguish the separable matrices among the positive definite matrices of the following type:

$$ Id\otimes Id+\sum_{i=1}^{l}a_i E_i\otimes F_i$$
where $tr(E_i)=tr(F_i)=0$, $\{E_1,...,E_l\},\ \{F_1,...,F_l\}$ are orthonomal sets of Hermitian matrices with respect to the trace inner product and $a_i\in\mathbb{R}$.

They also obtained a remarkable reduction of the separability problem in $M_{2}\otimes M_{2}$ for the general case, not only for the positive definite case. They  showed that, in order to solved it, it suffices to discover which matrices from the following family of matrices are separable.
$$ Id\otimes Id+d_2\gamma_2\otimes\gamma_2+d_3\gamma_3\otimes\gamma_3+d_4\gamma_4\otimes\gamma_4,$$
where $d_2,d_3,d_4\in\mathbb{R}$ and $\gamma_2,\gamma_3, \gamma_4$ are the matrices of the Pauli's basis of $M_{2}$ different from the $Id$.

They proved that a matrix in this family is separable if and only if is PPT, and if and only if $|d_2|+|d_3|+|d_4|\leq 1$.  
This is a second proof of the PPT criterion in  $M_{2}\otimes M_{2}$.

Another interesting kind of reduction is based on the concept of irreducible state developed (and defined) in \cite{chen1}, \cite{chen2}. The authors of these papers noticed that the separability problem can be reduced to the set of irreducible states (see corollary 16 in \cite{chen1}) and the states do not need to be positive definite. But in this paper there is no description of irreducible states. It is not possible to ensure if a state is irreducible or not with the ideas of these papers. Thus, this reduction is quite different from the reduction mentioned in the previous paragraph, since we do not have a description of the set of matrices where the separability problem was reduced. 

Our paper is an effort to generalize the result obtained in \cite{leinaas} for the general case (positive semidefinite case),  using a weaker concept of irreducible matrix (see definition \ref{definitionweakirreducible}).

We defined SPC matrices and our main results concern SPC matrices and PPT matrices (see definitions \ref{definitionSPC} and \ref{definitionPPT}).
We started to prove in section 2 that many results holds for both types. For example, if these matrices have specific Hermitian Schmidt decompositions then these matrices are separable in a very strong sense (see theorem \ref{theomiracle1} and corollary \ref{corollarytheomiracle1}). We proved that both types have what we call \textbf{split decompositions} (see theorem \ref{splitdecompositionSPC} and \ref{splitdecompositionPPT}). 

These theorems together with the notion of weak irreducible matrix, in section 3, imply that these matrices are weak irreducible or a sum of weak irreducible matrices of the same type (see theorem \ref{theoremsumofweakirreducible}) and the separability problem for these types of matrices can be reduced to the set of weak irreducible matrices with the same type (see proposition \ref{reductionfinal}). Notice that since a necessary condition for separability is to be PPT, we obtained quite a general reduction.

In a different manner as that of  the papers \cite{chen1} and \cite{chen2} and similarly to the paper \cite{leinaas}, we can provide a complete description of the set where we reduced the separability problem, in our case the set of weak irreducible SPC/PPT matrices (see theorem \ref{conditionsirreducible2}).  We discovered a condition based on the format of the Hermitian Schmdit decomposition of a positive semidefinite matrix that is sufficient for weak irreducibility(see theorem \ref{conditionsirreducible}). However, if the matrix is SPC or PPT, then the condition is also necessary (see theorem \ref{conditionsirreducible2}). It is interesting to notice that the family of positive definite matrices in the reduction obtained in \cite{leinaas} are all weak irreducible, because of corollary \ref{corollarylemmainequality} and also because satisfy the condition in theorem \ref{conditionsirreducible}.

In section 4, we showed that every positive semidefinite Hermitian matrix with tensor rank 2 has minimal separable decomposition, therefore it is separable (we made an application of this minimality to prove a similar result for the multipartite case, see theorem \ref{extension}). We shall use this result in section 5 to obtain two sharp inequalities that provide separability for SPC/PPT matrices. Again, since a necessary condition for separability is to be PPT, we obtained  quite a general inequality. These inequalities are generalizations of the inequality $|d_2|+|d_3|+|d_4|\leq 1$ mentioned before.

Finally, in our \textit{preliminary results}, we explict some ideas that are recurrent is this paper.  A recurrent need is the use of what we call $*-product$. Quite a few times, we shall use some properties of this $*-product$, for example, the $*-product$ of two positive semidefinite Hermitian matrices is also a positive semidefinite Hermitian matrix. We shall also use the minimal Hermitian decomposition and the Hermitan Schmidt decomposition of a Hermitian matrix, because these decompositions simplify our calculations, but they are also important for the description of weak irreducible matrices as mentioned before. For this purpose, we wrote quick subsections about these decompositions. In these subsections, we brought to the attention of the reader  a somewhat old result (theorem \ref{theoremprincipal})  and the definitions of the supports of a Hermitian matrix (see definition \ref{supp}), these two results are employed several times.
We tried to leave this paper self-contained, thus many of these \textit{preliminary results} are, therefore, not original.\\

\textbf{Notation:}  
The tensor product space of the vector spaces V, W over the field $\mathbb{C}$ will be denoted by $V\otimes W$.
We identify the tensor produt space $\mathbb{C}^n\otimes\mathbb{C}^k$ with $\mathbb{C}^{nk}$ and the tensor product space of complex matrices $M_{k}\otimes M_{m}$ with $M_{km}$, via Kronecker product. These identifications allow us to write $(v\otimes w)(r\otimes s)^t= vr^t\otimes ws^t$, where $v\otimes w$ is a column and $(v\otimes w)^t$ its transpose.  Therefore if $v,w\in\mathbb{C}^n\otimes\mathbb{C}^m$ we have $vw^t\in M_n\otimes M_m$.
The trace of a matrix $A$ is denoted by $tr(A)$ and $A^t$ is the transpose of $A$. 
\section{Preliminary Results}

We start this section with a subsection about the $*-product$ in $\mathbb{C}^n\otimes\mathbb{C}^m$ and in $M_{n}\otimes M_{m}$. These two products are examples of a common construction in tensor product spaces. If we consider $a\otimes b$ and $c\otimes d$, we can define a product $a\otimes b \times c\otimes d= f(b,c)\ a\otimes d$, where $f$ is a bilinear functional. From here, we can define for arbitrary tensors by means of distributivity.

 The first product is like the matrix product but for tensors in $\mathbb{C}^n\otimes\mathbb{C}^m$. The second one is a generalization of the Schur product. 

These products are very useful for research, notice how easily  proposition \ref{choitheorem} and lemma \ref{lemmapositive} were obtained. Actually, everytime we needed an argument about completely positive maps, we replaced it by a multiplication for a suitable positive semidefinite matrix.

Then we discuss some types of Hermitian decompositions: The minimal Hermitian decomposition and the Hermitian Schmidt decomposition. These decompositions simplify our calculations with the $*-product$ in $M_{n}\otimes M_{m}$ but they are also  important for another reason. We actually obtained a description of weak irreducible matrices, a concept that were exploited in section 3, based on the format of the Hermitian Schimidt decompositions of the matrices involved (see theorem \ref{conditionsirreducible2}).

\subsection{The Generalized Schur Product: $*-product$}
\indent\\

For the sake of completeness, in this subsection we proved that $*-product$ is distributive, associative, has an identity element and is well defined. This identity element plays an important role in Choi's theorem (see proposition \ref{choitheorem}).

\begin{definition}$($ The $*-product$ in $\mathbb{C}^n\otimes\mathbb{C}^m$ $)$: Let $v\in\mathbb{C}^n\otimes\mathbb{C}^m$ and $w\in\mathbb{C}^m\otimes\mathbb{C}^l$. Let $\displaystyle v=\sum_{i=1}^k v_i\otimes r_i,\ w=\sum_{j=1}^t w_j\otimes s_j$. Define $\displaystyle v*w=\sum_{i,j}v_i\otimes s_jtr(w_jr_i^t)$.\\
\end{definition}

\begin{proposition} \label{defu}   Let $\{e_i|1\leq i\leq m\}$ be the canonical basis of $\mathbb{C}^m$ and $\{v_i|1\leq i\leq m\}$ any other orthonormal basis of $\mathbb{C}^m$. We then have:
\begin{enumerate}
	\item $\displaystyle u=\sum_{i=1}^m e_i\otimes e_i=\sum_{i=1}^m v_i\otimes \overline{v_i}$ in $\mathbb{C}^m\otimes\mathbb{C}^m$.
	\item $u=\sum_{i=1}^m v_i\otimes \overline{v_i}$ is the identity in $\mathbb{C}^m\otimes\mathbb{C}^m$ with respect to the $*-product$.
\end{enumerate} 
\end{proposition}
\begin{proof} Consider the isomorphism $T:M_m\rightarrow \mathbb{C}^m\otimes\mathbb{C}^m$
\begin{center}
$T(\sum_{i=1}^k v_iw_i^t)=\sum_{i=1}^k v_i\otimes w_i$.
\end{center} 
Notice that $T^{-1}(\sum_{i=1}^m e_i\otimes e_i)=Id=T^{-1}(\sum_{i=1}^m v_i\otimes \overline{v_i})$, therefore item 1 is proved, since $T$ is an isomorphism.

Notice also that $T(AB)=T(A)*T(B)$. Therefore  $T(B)=T(Id B)=T(Id)*T(B)=u*T(B)$, and item 2 is proved.
\end{proof}

\begin{definition}$($The $*-product$ in $M_{n}\otimes M_{m}$ $)$:\  Let $A\in M_{n}\otimes M_{m}$ and $B\in M_{m}\otimes M_{l}$.\\ Let $\displaystyle A=\sum_{i=1}^k A_i\otimes D_i,\ B=\sum_{j=1}^t B_j\otimes C_j$. Define $\displaystyle A*B=\sum_{i,j}A_i\otimes C_jtr(D_iB_j^t)$.\\
\end{definition}

\begin{proposition} \label{rank1product} Let $v,w\in\mathbb{C}^n\otimes\mathbb{C}^m$ and $r,s\in\mathbb{C}^m\otimes\mathbb{C}^k$. The $*-product$ of the rank 1 matrices $vw^t$ and $rs^t$ is $vw^t*rs^t=(v*r)(w*s)^t$.
\end{proposition}
\begin{proof} Let $\{e_i|1\leq i\leq m\}$ be the canonical basis of $\mathbb{C}^m$. Write 
\begin{center}
$\displaystyle v=\sum_{i=1}^mv_i\otimes e_i,\ \ w=\sum_{j=1}^m w_j\otimes e_j,\ \ r=\sum_{k=1}^m e_k\otimes r_k,\ \ s=\sum_{l=1}^me_l\otimes s_l$.
\end{center}
Therefore $$\displaystyle vw^t=\sum_{i,j}v_iw_j^t\otimes e_ie_j^t, \ \ rs^t=\sum_{k,l} e_ke_l^t\otimes r_ks_l^t, \ \ vw^t*rs^t= \sum_{i,j}v_iw_j^t\otimes r_is_j^t.$$
Now $\displaystyle (v*r)(w*s)^t=(\sum_{i=1}^mv_i\otimes r_i)(\sum_{j=1}^mw_j\otimes s_j)^t=\sum_{i,j}v_iw_j^t\otimes r_is_j^t.$
\end{proof}

\begin{proposition}\label{defuut} Let $u$ be the vector defined in proposition $\ref{defu}$. The rank one matrix $u\overline{u}^t$ is the identity
in $M_{m}\otimes M_{m}$ with respect to the $*-product$.
\end{proposition}

\begin{proof} Notice that $u=\overline{u}$. This proposition follows directly from propositions \ref{defu} and \ref{rank1product}. 
\end{proof}

It is easy to notice that the $*-product$ in $M_{n}\otimes M_{m}$ is distributive, to prove the associativity of this product we use the following familiar maps.

\begin{definition}\label{defmaps} Let $A\in M_{n}\otimes M_{m}$ and $A=\displaystyle\sum_{i=1}^k A_i\otimes D_i$.
Define 
\begin{center}
$F_A:M_{m}\rightarrow M_{n}\hspace{3cm} G_A:M_{ n}\rightarrow M_{ m}$\vspace{0,25cm}\\ 
$F_A(Y)=\sum_{i=1}^k tr(YD_i)A_i\hspace{1,5 cm}G_A(X)=\sum_{i=1}^k tr(XA_i)D_i$
\end{center}
\end{definition}

\begin{remark}Every decomposition of A in $M_{n}\otimes M_{m}$ provides the same $F_A$ and $G_A$.
\end{remark}

The next proposition connects these maps to the $*-product$ and shows that all decompositions used for $A$ or $B$ to compute $A*B$ provide the same $A*B$ (therefore this product is well defined).

\begin{proposition}\label{AB} $A*B=F_A((\cdot)^t)\otimes Id (B)= Id\otimes G_B((\cdot)^t) (A)$
\end{proposition}

\begin{proof} Let $A=\displaystyle\sum_{i} A_i\otimes D_i$ and $\displaystyle B=\sum_{j} B_j\otimes C_j$. \\
Therefore 
$\displaystyle A*B=\sum_{i,j}A_i\otimes C_jtr(D_iB_j^t)=\sum_{j}(\sum_{i}tr(B_j^tD_i)A_i)\otimes C_j=\sum_{j}F_A(B_j^t)\otimes C_j=F_A((\cdot)^t)\otimes Id (B)$. But $A*B$ is also equal to 
$\displaystyle\sum_{i}A_i\otimes (\sum_{j}tr(B_jD_i^t)C_j)=\sum_{i} A_i\otimes G_B(D_i^t)=Id\otimes G_B((\cdot)^t) (A)$. 
\end{proof}

\begin{corollary} \label{corollaryAB} $A=F_A((\cdot)^t)\otimes Id (u\overline{u}^t)$, where $u\overline{u}^t$ is the matrix of proposition \ref{defuut}.
\end{corollary}

\begin{remark} The map $A\rightarrow F_A((\cdot)^t)$ is the inverse of Jamio􏰘lkowski's isomorphism \cite{Jamiol􏰘kowski}.
\end{remark}

\begin{proposition} \label{F_AB}$F_{A*B}(Y)=F_A(F_B(Y)^t)$ and $G_{A*B}(X)=G_B(G_A(X)^t)$
\end{proposition}
\begin{proof}Let $\displaystyle A=\sum_{i} A_i\otimes D_i$ and $\displaystyle B=\sum_{j} B_j\otimes C_j$. \\
By proposition \ref{AB},
$\displaystyle A*B=F_A((\cdot)^t)\otimes Id (B)=\sum_{j}F_A(B_j^t)\otimes C_j.$\\
Thus, $\displaystyle F_{A*B}(Y)=\sum_{j}F_A(B_j^t)tr(YC_j)=F_A((\sum_{j}B_j tr(YC_j))^t)=F_A(F_B(Y)^t)$.\\
The proof of $G_{A*B}(X)=G_B(G_A(X)^t)$ is analogous.
\end{proof}

\begin{corollary} The $*-product$ is associative.
\end{corollary}
\begin{proof}
$(A*B)*C=F_{A*B}((\cdot)^t)\otimes Id(C)=(F_A((\cdot)^t)\otimes Id)\circ (F_B((\cdot)^t)\otimes Id)(C)=A*(B*C).$
\end{proof}

\subsection{Interesting Consequences of $*-$Product}
\indent\\

The next proposition and its corollary shows that the $*-product$ of two positive semidefinite Hermitian matrices is a positive semidefinite Hermitian matrix. This fact is used many times in this article.

\begin{proposition}\label{important} Let $v\in\mathbb{C}^n\otimes\mathbb{C}^m$ and $w\in\mathbb{C}^m\otimes\mathbb{C}^l$. Then $v\overline{v}^t*w\overline{w}^t=(v*w)\overline{(v*w)}^t$.
\end{proposition}
\begin{proof} Let $\{e_i|1\leq i\leq m\}$ be the canonical basis for $\mathbb{C}^m$.
We can write $v=\sum_{i=1}^{m}v_i\otimes e_i$ and $w=\sum_{j=1}^{m}e_j\otimes w_j$, 
Therefore $$v\overline{v}^t=\sum_{i,k} v_i\overline{v_k}^t\otimes e_ie_k^t,\  w\overline{w}^t=\sum_{j,s} e_je_s^t\otimes w_j\overline{w_s}^t.$$
Thus $v\overline{v}^t*w\overline{w}^t=\sum_{i,k} v_i\overline{v_k}^t\otimes w_i\overline{w_k}^t$. \\
Notice that $v*w=\sum_{i=1}^{m}v_i\otimes w_i$ and $(v*w)\overline{(v*w)}^t=\sum_{i,k} v_i\overline{v_k}^t\otimes w_i\overline{w_k}^t$.
\end{proof}

\begin{corollary}\label{corollarypositive} If $A\in M_{n}\otimes M_{m}$ and $B\in M_{m}\otimes M_{l}$ are positive semidefinite Hermitian matrices then $A*B$ is also.
\end{corollary}
\begin{proof} By hypothesis $A=\sum_iv_i\overline{v_i}^t$ and $B=\sum_jw_j\overline{w_j}^t$, for $v_i\in\mathbb{C}^n\otimes\mathbb{C}^m$ and $w_j\in\mathbb{C}^m\otimes\mathbb{C}^l$.
Therefore $$A*B=\sum_iv_i\overline{v_i}^t*\sum_jw_j\overline{w_j}^t=\sum_{i,j}v_i\overline{v_i}^t*w_j\overline{w_j}^t$$
By proposition \ref{important},  
$A*B=\sum_{i,j}(v_i*w_j)\overline{(v_i*w_j)}^t$.
\end{proof}

\begin{proposition}\label{carac} If for every positive semidefinite Hermitian matrix $B\in M_{ m}\otimes M_{ l}$ $($for arbitrary $l$ $)$, $A*B$ is also a positive semidefinite Hermitian matrix,  then A is a positive semidefinite Hermitian matrix.
\end{proposition}
\begin{proof}The matrix $u\overline{u}^t$ from proposition \ref{defuut} is Hermitian and positive semidefinite. Thus, $A*u\overline{u}^t=A$ must be Hermitian and positive semidefinite.
\end{proof}

An easy consequence of this proposition is the characterization of all completely positive maps, obtained in \cite{choi}.

\begin{definition} A linear transformation  $T:M_{m}\rightarrow M_{ n}$ is completely positive if for every positive semidefinite Hermitian matrix $B\in M_{ m}\otimes M_{ l}$, we obtain a positive semidefinite Hermitian matrix $T\otimes Id(B)$. 
\end{definition}

\begin{proposition}\label{choitheorem} \textbf{$($Choi's Theorem$)$} A linear transformation $T:M_{m}\rightarrow M_{ n}$ is a completely positive map if and only if $T\otimes Id(u\overline{u}^t)$ is a positive semidefinite Hermitian matrix.
\end{proposition}
\begin{proof} Let $A=T\otimes Id(u\overline{u}^t)=\sum_{i,j=1}^m T(e_ie_j^t)\otimes e_ie_j^t$. Now $$F_A(Y)=\sum_{i,j=1}^m T(e_ie_j^t)tr(Ye_ie_j^t)=T(\sum_{i,j=1}^m e_ie_j^t tr(Ye_ie_j^t))=T(Y^t).$$

Notice that $A*B=F_A((\cdot)^t)\otimes Id(B)=T\otimes Id(B)$. Thus $T$ is completely positive if and only if $A*B$ is a positive semidefinite Hermitian matrix for every positive semidefinite Hermitian matrix $B$. Therefore $A$ must be a positive semidefinite Hermitian matrix (by proposition \ref{carac}).
\end{proof}

\subsection{Minimal Hermitian Decomposition}
\indent\\

In this subsection we present some definitions and give a quick proof (alternative proof) of the fact that every Hermitian matrix has a minimal Hermitian decomposition, this result can also be found in \cite{shao}. Based on this result, in the next  subsection, we also give a quick proof (alternative proof) that every Hermitian matrix has a Hermitian Schmidt decomposition, this result can also be found in \cite{herbut}. 
These decompositions simplify the calculations with $*-product$ and
in section 3, theorem \ref{conditionsirreducible2}, we give a complete description of the weak irreducible SPC/PPT matrices based on a property of their Hermitian Schmidt decompositions.\\

Here, and from now on, the subspace generated by $\{m_1,...,m_t\}$ will be denoted by $(m_1,...,m_t)$. 

\begin{definition} \label{definitiontensorrank}Let $r \in V\otimes W$. The tensor rank of $r$ is $1$, if $r=v\otimes w$ and $r\neq0$. The tensor rank of $r$ is the minimal number of tensors with tensor rank 1 that can be added to form $r$.
\end{definition}

\begin{theorem}\label{theoremprincipal}
Let V and W be vector spaces over the field F and let  $v_i,r_j\in V$ and $w_i,s_j\in W$, for $1\leq i\leq n$ and $1\leq j\leq k$.\\ Let $\sum_{i=1}^n v_i\otimes w_i=\sum_{j=1}^kr_j\otimes s_j\in V\otimes_F W$.
\begin{itemize}
 \item [a)] If $\{v_1,...,v_n\}$ is a linear independent set then $(w_1,...,w_n)\subset(s_1,...,s_k)$.
 \item [b)] If $\{w_1,...,w_n\}$ is a linear independent set then $(v_1,...,v_n)\subset(r_1,...,r_k)$.
\end{itemize}

\end{theorem}
\begin{proof} See reference \cite{marvin}.
\end{proof}

\begin{corollary}\label{corollarytensorrank}Let $\sum_{i=1}^n v_i\otimes w_i=\sum_{j=1}^kr_j\otimes s_j$.\\ If $\{v_1,...,v_n\}$ and $\{w_1,...,w_n\}$ are linear independent sets then $k\geq n$. So the tensor rank of  $\sum_{i=1}^n v_i\otimes w_i$ is n.
\end{corollary}
\begin{proof} See reference \cite{marvin}.
\end{proof}

\begin{definition}
A decomposition $\sum_{i=1}^n A_i\otimes B_i$ of a matrix A, is said to be a minimal decomposition, if $\{A_1,...,A_n\}$ and $\{B_1,...,B_n\}$ are linear independent sets. This nomenclature is justified by corollary \ref{corollarytensorrank}.
\end{definition}

\begin{definition} A decomposition $\sum_{i=1}^n A_i\otimes B_i$ of a Hermitian matrix A is a Hermitian decomposition if $\{A_1,...,A_n\}$ and $\{B_1,...,B_n\}$ are Hermitian matrices. Also, if  $\{A_1,...,A_n\}$ and $\{B_1,...,B_n\}$ are linear independent sets then $\sum_{i=1}^n A_i\otimes B_i$ is a minimal Hermitian decomposition of A. 
\end{definition}

 Every Hermitian matrix has a minimal Hermitian decomposition. This result was proved in \cite{shao} but, here below, we present a quick proof for the convenience of the reader.

\begin{lemma} Every Hermitian matrix $A_{km\times km}$ has a Hermitian decomposition in $M_{k}\otimes M_{m}$.
\end{lemma}
\begin{proof} Let $A_1\otimes B_1+...+A_n\otimes B_n$ be any decomposition of $A$ in\\ $M_{k}\otimes M_{m}$.  Let  $A_j=H_1^j+iH_2^j$ and $B_j=S_1^j+iS_2^j$, where $H_k^j$ and $S_k^j$ are Hermitian matrices for every $k$ and $j$.  So $A=\sum_{j=1}^n A_j\otimes B_j=$
\begin{center} $\sum_{j=1}^nH_1^j\otimes S_1^j+\sum_{j=1}^n(-H_2^j)\otimes S_2^j +i(\sum_{j=1}^nH_1^j\otimes S_2^j+\sum_{j=1}^nH_2^j\otimes S_1^j)$.
\end{center}
Notice that $i(\sum_{j=1}^nH_1^j\otimes S_2^j+\sum_{j=1}^nH_2^j\otimes S_1^j)$ is anti-hemitian matrix, because $\sum_{j=1}^nH_1^j\otimes S_2^j+\sum_{j=1}^nH_2^j\otimes S_1^j$ is a Hermitian matrix. But 
\begin{center}
$i(\sum_{j=1}^nH_1^j\otimes S_2^j+\sum_{j=1}^nH_2^j\otimes S_1^j)= A- (\sum_{j=1}^nH_1^j\otimes S_1^j+\sum_{j=1}^n(-H_2^j)\otimes S_2^j)$,
\end{center}
 So $i(\sum_{j=1}^nH_1^j\otimes S_2^j+\sum_{j=1}^nH_2^j\otimes S_1^j)$ is also Hermitian as a difference of two Hermitian matrices. We obtain  $i(\sum_{j=1}^nH_1^j\otimes S_2^j+\sum_{j=1}^nH_2^j\otimes S_1^j)=0$.
 
Finally $A=\sum_{j=1}^nH_1^j\otimes S_1^j+\sum_{j=1}^n(-H_2^j)\otimes S_2^j$, which is a Hermitian decomposition for $A$.

\end{proof}

\begin{lemma}\label{real} Let A be a Hermitian matrix and $\{A_1,...,A_n\}$ be a linear independent set of Hermitian matrices. If $A=\sum_{k=1}^n z_k A_k$ then $z_k\in\mathbb{R}$ for every k. 

\end{lemma}

\begin{proof}$A=\overline{A^t}=\sum_{k=1}^n\overline{z_k}\overline{A_k^t}=\sum_{k=1}^n\overline{z_k}A_k$. Since $\{A_1,...,A_n\}$ is a linear independent set we have $\overline{z_k}=z_k$.
\end{proof}

\begin{lemma}\label{minimalHermitiandecomposition}If  A is a Hermitian matrix in $M_{k}\otimes M_{m}$ then A has a minimal Hermitian decomposition.
\end{lemma}
\begin{proof} Let $\sum_{i=1}^n A_i\otimes B_i$ be a Hermitian decomposition of A. If $\{A_1,...,A_n\}$ and $\{B_1,...,B_n\}$ are linear independent sets then $\sum_{i=1}^n A_i\otimes B_i$ is a minimal Hermitian decomposition. If not suppose that $\{A_1,...,A_n\}$ is a linear dependent set. Extract from $\{A_1,...,A_n\}$ a basis for $(A_1,..,A_n)$. Without loss of generality suppose $\{A_1,...,A_k\}$ is this basis. Write $A_n=\sum_{i=1}^k z_iA_i$. By lemma \ref{real}, $z_i\in\mathbb{R}$. So $\sum_{i=1}^n A_i\otimes B_i=\sum_{i=1}^{k} A_i\otimes (B_i+z_iB_n)+\sum_{i=k+1}^{n-1} A_i\otimes B_i$. Notice that $B_i+z_iB_n,\ 1\leq i\leq k$, are Hermitian matrices. Therefore we got a smaller Hermitian decomposition for $A$. Repeat the argument until you find a minimal Hermitian decomposition.
\end{proof}

\subsection{Hermitian Schmidt Decomposition}
\indent\\

If you consider the inner product in $M_{m}$ as $\langle C,D\rangle=\mbox{trace}(CD^*)$, we can define a Schmidt decomposition of a matrix $A_{km\times km}$.

\begin{definition} A decomposition of a matrix A in $ M_{k}\otimes M_{m}$,
$$\sum_{i=1}^n \lambda_i \gamma_i\otimes \delta_i$$
is a Schmidt decomposition if $\{\gamma_i|1\leq i\leq n\}\subset M_k$,\ $\{\delta_i|1\leq i\leq n\}\subset M_m$ are orthonormal sets,  $\lambda_i\in\mathbb{R}$ and $\lambda_i>0$.
Also, if $\gamma_i$ and $\delta_i$ are Hermitian matrices for every $i$, then $\sum_{i=1}^n \lambda_i \gamma_i\otimes \delta_i$ is a Hermitian Schmidt decomposition of A.
 Notice that every Schmidt decomposition is minimal.
\end{definition}
 
Next we present a quick proof of the existence of such decomposition for every Hermitian matrix in $M_{k}\otimes M_{m}$ for the convenience of the reader, another proof can be found in \cite{herbut}. First we need to define the supports of a matrix. The definition of the supports was used many times in this paper. \\

\begin{definition} \label{supp}
Let $\sum_{i=1}^n A_i\otimes B_i$ be a minimal Hermitian decomposition of $A_{mk\times mk}$ on $M_{k}\otimes M_{m}$. Define the supports of $A$ as\\
$supp_1(A)=\mbox{the real subspace generated by the Hermitian matrices}\ A_1,...,A_n$,\\  
$supp_2(A)=\mbox{the real subspace generated by the Hermitian matrices}\  B_1,...,B_n.$

Notice that the matrices in $supp_1(A)$ and $supp_2(A)$ are Hermitians. These supports do not depend on the minimal Hermitian decomposition of A, because of theorem \ref{theoremprincipal} and lemma \ref{real}.\\
\end{definition}

\begin{lemma}\label{adjoint}Let $A\in M_{k}\otimes M_{m}$ then  $tr(F_A(Y)X)=tr(YG_A(X))$, where $F_A$ and $G_A$ are the maps of definition \ref{defmaps}.
\end{lemma}
\begin{proof}
$tr(F_A(Y)X)=tr(\sum_{i=1}^ntr(YB_i)A_iX)=\sum_{i=1}^ntr(YB_i)tr(A_iX)=$ 

$=tr(Y\sum_{i=1}^nB_itr(A_iX))=tr(Y\sum_{i=1}^nB_itr(XA_i))=tr(YG_A(X))$.
\end{proof}

\begin{remark}\label{remarkadjoint}
If we restrict the domain and the codomain of $F_A$ and $G_A$ to the supports of $A$,
$$F_A:supp_2(A)\rightarrow supp_1(A)\hspace{0.5cm}G_A:supp_1(A)\rightarrow supp_2(A),$$ these maps become real linear tranformations. The inner product restricted to the set of Hermitian matrices is only the trace, therefore by lemma \ref{adjoint}, these are now adjoints maps.
\end{remark}

\begin{theorem}\label{SVD} Every Hermitian matrix $A_{km\times km}$ has a Hermitian Schmidt decomposition on $M_{k}\otimes M_{m}$.
\end{theorem}
\begin{proof} Let $\sum_{i=1}^n A_i\otimes B_i$ be a minimal Hermitian decomposition of $A_{mk\times mk}$. Since $F_A$ and $G_A$ are adjoints,  $F_A\circ G_A:supp_1(A)\rightarrow supp_1(A)$ is self-adjoint. Thus, there exists an orthonormal basis for $supp_1(A)$, formed by eigenvectors of $F_A\circ G_A$. 
Let this basis be $\gamma=\{\gamma_1,...,\gamma_n\}$. 

Thus $A=\gamma_1\otimes\delta_1+...+\gamma_n\otimes\delta_n$, where $\delta_i=G_A(\gamma_i)$. 

Since $\gamma_i,\gamma_j$ are orthogonal eigenvectors of $F_A\circ G_A$, 
\begin{center}
$\langle\delta_j,\delta_i\rangle=\langle G_A(\gamma_j),G_A(\gamma_i)\rangle =\langle\gamma_j,F_A\circ G_A(\gamma_i)\rangle=0.$
\end{center}
Thus $|\delta_1|\gamma_1\otimes\frac{\delta_1}{|\delta_1|}+...+|\delta_n|\gamma_n\otimes\frac{\delta_n}{|\delta_n|}$ is a Hermitian Schmidt decomposition of A.
\end{proof}

\begin{remark}
\end{remark}
This proof is analogous to the proof of SVD decomposition of a linear transformation (see page 204, reference \cite{simon}), although SVD decomposition does not garantee the hermiticity of the decomposition.

\section{Split Decompositions for SPC and PPT Matrices}

We begin this section defining SPC and PPT matrices considered in the main theorems of this paper. Then we  prove some results that holds for both types. Actually, some theorems concerning  PPT matrices are consequences of the theorems obtained for SPC matrices. 

The main results of this section are theorems \ref{theomiracle1} and corollary \ref{corollarytheomiracle1}, theorems \ref{splitdecompositionSPC} and \ref{splitdecompositionPPT}.
 The first two results actually show that under certain hypothesis on the Hermitian Schmidt decomposition, the SPC/PPT matrices are separable in a very strong sense (see definition \ref{definitionseparability}).
The last two results are called the split decompositions of the SPC/PPT matrices.

In section 3, we continue to obtain theorems that holds for both types concerning the weak irreducible property and we noticed that SPC/PPT matrices that have trivial split decomposition are in fact weak irreducible.

In our final result (in section 5), we provide a sharp inequalities that provide separability for SPC/PPT matrices. 
Since a necessary condition for separability of a matrix $A$ is to be PPT, we obtain quite a general result.

 \begin{definition}  \label{definitionSPC}\textbf{$($SPC matrices$)$}Let $A\in M_k\otimes M_k$ be a positive semidefinite Hermitian matrix. We say that $A$  is symmetric with positive coefficients or simply SPC, if $A$ has the following Hermitian Schmidt decomposition: $A=\sum_{i=1}^n\lambda_i\gamma_i\otimes\gamma_i$ such that $\lambda_i>0$ for every $i$.
\end{definition}

\begin{definition} \label{definitionPPT}\textbf{$($PPT matrices$)$}Let $A\in M_k\otimes M_m$ be a positive semidefinite Hermitian matrix. 
We say that $ A=\sum_{i}A_i\otimes B_i$ is positive under partial transposition or simply PPT, if 
$A^{t_1}=(\cdot)^t\otimes Id(A)=\sum_{i}A_i^t\otimes B_i$ still is a positive semidefinite Hermitian matrix.\\
\end{definition}

To obtain our first theorem regarding SPC matrices (theorem \ref{theomiracle1}), we need the following lemmas. 

\begin{lemma}\label{new3}
Let $A\in M_k\otimes M_m$ be a positive semidefinite Hermitian matrix. Let $\sum_{i=1}^n\lambda_i\gamma_i\otimes\delta_i$ be a Hermitian Schmidt decomposition of $A$. If  for some $j$, the matrices $\gamma_j$ and $\delta_j$ are positive semidefinite then 
$$[\Im(\gamma_j)\otimes\ker(\delta_j)]\oplus[\ker(\gamma_j)\otimes\Im(\delta_j)]\subset\ker(A).$$

\end{lemma}
\begin{proof}

Let $w\in \ker(\delta_j)$ and let $\sum_ia_iv_i\overline{v_i}^t$ be a spectral decomposition of $\gamma_j$. Thus $a_i>0$.
Notice that $tr(A(\gamma_j\otimes w\overline{w}^t))=\lambda_j tr(\delta_j w\overline{w}^t)=0$, therefore $$0=tr(A(\gamma_j\otimes w\overline{w}^t))=\sum_ia_itr(A(v_i\overline{v_i}^t\otimes w\overline{w}^t)).$$
Since $a_i>0$ and $tr(A(v_i\overline{v_i}^t\otimes w\overline{w}^t))\geq 0$ we got $tr(A(v_i\overline{v_i}^t\otimes w\overline{w}^t))= 0$ for every $i$. Therefore $v_i\otimes w\in\ker(A)$.

Now the eigenvectors of $\gamma_j$ associated to the eigenvalues $a_i>0$, i.e., the vectors $v_i$, span  $\Im(\gamma_j)$. Then $\Im(\gamma_j)\otimes\ker(\delta_j)\subset\ker(A)$.

To obtain $\ker(\gamma_j)\otimes\Im(\gamma_j)\subset\ker(A)$, the argument is analogous.
\end{proof}

\begin{lemma}\label{new5} Let $A\in M_k\otimes M_k$ be a SPC matrix. 
If for some $\gamma_j$ in the Hermitian Schmidt decomposition of $A$, $\gamma_j$ is a positive semidefinite Hermitian matrix  and if

\begin{enumerate}
\item $r\in\Im(\gamma_j)$ and $s\in\ker(\gamma_j)$ or
\item $r\in\ker(\gamma_j)$ and $s\in\Im(\gamma_j)$
\end{enumerate}
 then  $tr(\gamma_i (r\overline{s}^t))=0$ for every $1\leq i\leq n$. 
\end{lemma}
\begin{proof} 

Let $r\in\Im(\gamma_j)$, $s\in\ker(\gamma_j)$, therefore $r\otimes s \in \Im(\gamma_j)\otimes\ker(\gamma_j)$. 

Now since $\gamma_j$ is a positive semidefinite Hermitian matrix, by lemma \ref{new3}, $r\otimes s\in \ker(A)$.

Therefore, $0=tr(A(r\otimes s)\overline{(s\otimes r )}^t)=tr(A(r\overline{s}^t\otimes s\overline{r}^t))$ 
$$=\sum_{i=1}^n\lambda_i tr(\gamma_ir\overline{s}^t)tr(\gamma_i s\overline{r}^t)=\sum_{i=1}^n\lambda_i tr(\gamma_ir\overline{s}^t)\overline{tr(\gamma_ir\overline{s}^t)}$$
Since $\lambda_i>0$ and $tr(\gamma_ir\overline{s}^t)\overline{tr(\gamma_ir\overline{s}^t)}\geq0$ we obtain $tr(\gamma_ir\overline{s}^t)=0.$

For $r\in\ker(\gamma_j)$ and $s\in\Im(\gamma_j)$ the proof is analogous.

\end{proof}

\begin{corollary}\label{corollarynew5}  Let $A\in M_k\otimes M_k$ be a SPC matrix and  $\gamma_j\in M_k$ as in lemma \ref{new5}. 
Let $V_1\in M_k$ be the Hermitian projection onto $\Im(\gamma_j)$ and $V_2\in M_k$ be the Hermitian projection onto $\ker(\gamma_j)$. Then
$A=\sum_{i=1}^2(V_i\otimes V_i) A (V_i\otimes V_i)$.
\end{corollary}
\begin{proof}
Recall that $A\in M_k\otimes M_k$ and $\gamma_j\in M_k$ are positive semidefinite and $\sum_{i=1}^n\lambda_i\gamma_i\otimes\gamma_i$ is a Hermitian Schmidt decomposition of $A$.

By lemma \ref{new3}, $A(V_i\otimes V_j)=0$ for $i\neq j$. Thus, $[A(V_i\otimes V_j)]^*=(V_i\otimes V_j)A=0$ for $i\neq j$ too. Since $V_1+V_2=Id\in M_k$, we obtain $A=\sum_{i,j=1}^2(V_i\otimes V_i) A (V_j\otimes V_j).$

Now, write $V_1=\sum_{m} r_m\overline{r_m}^t$ and $V_2=\sum_{l} s_l\overline{s_l}^t$, where the vectors $r_m$ form an orthonormal basis of $\Im(\gamma_j)$ and  the vectors $s_l$ form an orthonormal basis of $\ker(\gamma_j)$.

Now $\displaystyle V_1\gamma_i V_2=(\sum_{m}r_m\overline{r_m}^t)\gamma_i (\sum_{l}s_l\overline{s_l}^t)=\sum_{m,l} (r_m\overline{r_m}^t)\gamma_i (s_l\overline{s_l}^t)=$ 

\hspace{2,2cm}$\displaystyle =\sum_{m,l} (r_m\overline{s_l}^t)(\overline{r_m}^t\gamma_i s_l)=\sum_{m,l} (r_m\overline{s_l}^t) tr(\gamma_i s_l\overline{r_m}^t).$

Next, by item 2 of lemma \ref{new5}, $tr(\gamma_i s_l\overline{r_m}^t)=0$ for every $m$ and $l$, since $s_l\in ker(\gamma_j)$ and $r_m\in\Im(\gamma_j)$ and $\gamma_j$ is positive semidefinite.

So $V_1\gamma_iV_2=0$ for every $i$ and thus $(V_1\gamma_iV_2)^*=V_2\gamma_iV_1=0$ too.

Finally $ A=\sum_{i=1}^2(V_i\otimes V_i) A (V_i\otimes V_i)$.

\end{proof}

\begin{lemma} \label{lemmaY}Let $\gamma,B$ be Hermitian matrices in $M_k$, $\gamma$ positive semidefinite and $B\neq 0$.
Suppose that $\Im(B)\subset\Im(\gamma)$ and $B$ is not a multiple of $\gamma$. 
Then, there exists $\lambda\in\mathbb{R}$ such that $\gamma-\lambda B$ is positive semidefinite, non-null and  $0\neq x\in\ker(\gamma-\lambda B)\cap \Im(\gamma)$.
\end{lemma}
\begin{proof}
Let $R\in M_k$ be an invertible matrix such that $$R\gamma R^{*}=\left(\begin{array}{cc}
Id_{s\times s} & 0_{s\times t}\\
0 _{t\times s}           & 0_{t\times t}
\end{array}\right).$$

Since $\Im(B)\subset\Im(\gamma)$. We obtain $$RBR^{*}=\left(\begin{array}{cc}
\widetilde{B}_{s\times s} & 0_{s \times t}\\
0 _{t\times s}           & 0_{t\times t}
\end{array}\right).$$

Let $\frac{1}{\lambda}$ be the eigenvalue of $\widetilde{B}$ with the biggest absolute value. Therefore $Id-\lambda \widetilde{B}$ is positive semidefinite with non-null kernel. Then 
$$\dim(\ker(\gamma-\lambda B))=\dim(\ker(R(\gamma-\lambda B) R^{*}))>\dim(\ker(R\gamma R^{*}))=\dim(\ker(\gamma)).$$

Thus, there exists $0\neq x\in \ker(\gamma-\lambda B)\cap \ker(\gamma)^{\perp}. $ Since $\gamma$ is Hermitian $\ker(\gamma)^{\perp}=\Im(\gamma).$

Since $R(\gamma-\lambda B)R^*$ is positive semidefinite then $\gamma-\lambda B$ is too.

Therefore exist a $\lambda\in\mathbb{R}$ such that $\gamma-\lambda B$ is positive semidefinite and $0\neq x\in \ker(\gamma-\lambda B)\cap \Im(\gamma). $
\end{proof}

\subsection{First Results for SPC and PPT Matrices}
\indent\\

Our first theorem regarding SPC matrices is the following one.

\begin{theorem}\label{theomiracle1}
Let $A\in M_k\otimes M_k$ be a SPC matrix with all $\lambda_i=1$. Then
\begin{itemize}
\item[a)] $A$ has a Hermitian Schmidt decomposition,
$A= \sum_{i=1}^n\gamma_i'\otimes\gamma_i'$  such that $\gamma_i'$ is a positive semidefinite Hermitian matrix for $1\leq i\leq n$.
\item[b)] The decomposition of item $a)$ is unique.
\end{itemize}
\end{theorem}
\begin{proof} 
Let $\sum_{i=1}^n\gamma_i\otimes\gamma_i$ be a Hermitian Schmidt decomposition of $A$.\\\\
a) The proof will be done by induction on the tensor rank of $A$.

 If $A=\gamma_1\otimes\gamma_1$ is positive semidefinite then it is obvious that $\gamma_1$ or $-\gamma_1$  is positive semidefinite, but $A=\gamma_1\otimes\gamma_1=-\gamma_1\otimes-\gamma_1$, and we are done.
Suppose $n$  bigger than 1.

 If $\sum_{i=1}^n tr(\gamma_i)\gamma_i=0$ then $tr(\gamma_i)=0$ for $1\leq i\leq n$, since the $\gamma_i$'s are orthonormal. It implies that $ tr(A)=\sum_{i=1}^n tr(\gamma_i)tr(\gamma_i)=0$ and since $A$ is positive semidefinite, we obtain $A=0$. This is a contradiction.

Since $A\neq0$ then $\sum_{i=1}^n tr(\gamma_i)\gamma_i\neq0$. Since $A*Id\otimes Id=\sum_{i=1}^n tr(\gamma_i)\gamma_i\otimes Id$ is positive semidefinite, we obtain that $\sum_{i=1}^n tr(\gamma_i)\gamma_i$ is Hermitian and positive semidefinite.

Define $\gamma_1'=\sum_{i=1}^n tr(\gamma_i)\gamma_i/|\sum_{i=1}^n tr(\gamma_i)\gamma_i|$. Notice that $\gamma_1'\in supp_1(A)$. 
 
Let $\{\gamma_1',\widetilde{\gamma_2},...,\widetilde{\gamma_n}\}$be any orthonormal basis of $supp_1(A)$ containing $\gamma_1'$.

Since $A=\sum_{i=1}^n\gamma_i\otimes\gamma_i$, we have $supp_1(A)=supp_2(A)$. Notice that $G_A:supp_1(A)\rightarrow supp_2(A)$ is the identity $Id:supp_1(A)\rightarrow supp_1(A)$.

Now 
$A=\gamma_1'\otimes G_A(\gamma_1')+\widetilde{\gamma_2}\otimes G_A(\widetilde{\gamma_2})+...+\widetilde{\gamma_n}\otimes G_A(\widetilde{\gamma_n})$
 $=\gamma_1'\otimes \gamma_1'+\widetilde{\gamma_2}\otimes \widetilde{\gamma_2}+...+\widetilde{\gamma_n}\otimes \widetilde{\gamma_n}.$

Now to simplify the notation, we may assume without loss of generality, that $A=\sum_{i=1}^n\gamma_i\otimes\gamma_i$ with $\gamma_1$ positive semidefinite.

Our aim now is to write $A=A_1+A_2$, where $A_1$ and $A_2$ have the same type of $A$, but both with smaller tensor rank and use induction on the tensor rank.
Next we split the proof in two cases:\\\\
\underline{First Case:} Exists $i>1$ such that $\Im(\gamma_i)$ is not contained in $\Im(\gamma_1)$.\\

Under this hypothesis $\ker(\gamma_1)\neq 0$. Let $V_1\in M_k$ be the Hermitian projection onto $\Im(\gamma_1)$ and $V_2\in M_k$ be the Hermitian projection onto $\ker(\gamma_1)$, therefore $V_1+V_2=Id_k $.

By corollary \ref{corollarynew5}, $A=\sum_{i=1}^2(V_i\otimes V_i) A (V_i\otimes V_i)$.

If $(V_2\otimes V_2) A (V_2\otimes V_2)= 0$ then $A=(V_1\otimes V_1) A (V_1\otimes V_1)$ and $supp_1(A)\subset
$ $\text{span}\{V_1\gamma_1V_1,\ldots, V_1\gamma_nV_1\}$, by theorem \ref{theoremprincipal}. Then every $\gamma_i\in supp_1(A)$ would have $\Im(\gamma_i)\subset\Im(V_1)=\Im(\gamma_1)$, which is a contradiction with the hypothesis of this case. Thus, $(V_2\otimes V_2) A (V_2\otimes V_2)\neq 0$.
 Now, if $(V_1\otimes V_1) A (V_1\otimes V_1)= 0$ then $A=(V_2\otimes V_2) A (V_2\otimes V_2)$  and again we obtain that $\Im(\gamma_1)\subset\Im(V_2)=\ker(\gamma_1)$. Since $\gamma_1$ is Hermitian  $\Im(\gamma_1)\perp\ker(\gamma_1)$, but $\gamma_1\neq 0$. It is a contradiction and $(V_1\otimes V_1) A (V_1\otimes V_1)\neq 0$.

Let us write $A=A_1+A_2$, where $A_i=(V_i\otimes V_i)A(V_i\otimes V_i)$. Notice that $supp_1(A_1)\perp supp_1(A_2)$, since $V_1V_2=0$.

Let $\delta_1,\ldots,\delta_r$ be an orthonormal basis of $supp_1(A_1)$ and $\epsilon_1,\ldots,\epsilon_s$ be an orthonormal basis of $supp_1(A_2)$. Since $supp_1(A)=supp_1(A_1)\oplus supp_1(A_2)$ then 
$A=\delta_1\otimes G_A(\delta_1)+\ldots+\delta_r\otimes G_A(\delta_r)+\epsilon_1\otimes G_A(\epsilon_1)+\ldots+\epsilon_s\otimes G_A(\epsilon_s)=\delta_1\otimes \delta_1+\ldots+\delta_r\otimes \delta_r+\epsilon_1\otimes \epsilon_1+\ldots+\epsilon_s\otimes \epsilon_s$. 

Thus, $A_1-\delta_1\otimes \delta_1-\ldots-\delta_r\otimes \delta_r=\epsilon_1\otimes \epsilon_1+\ldots+\epsilon_s\otimes \epsilon_s-A_2$ which implies  $supp_1(A_1-\delta_1\otimes \delta_1-\ldots-\delta_r\otimes \delta_r)\subset supp_1(A_1)\cap supp_1(A_2)=\{0\}$. Thus, $ A_1-\delta_1\otimes \delta_1-\ldots-\delta_r\otimes \delta_r=0$ and $\epsilon_1\otimes \epsilon_1+\ldots+\epsilon_s\otimes \epsilon_s-A_2=0$. 

Recall that both  $A_i\neq 0$ then $r+s=n$ and $r<n$ and $s<n$. Thus $A_1$ and $A_2$ are positive semidefinite Hermitian matrices with Hermitian Schmidt decompositions similar to $A$, but with smaller tensor rank.\\\\
\underline{Second Case:} For every $i$,  $\Im(\gamma_i)\subset \Im(\gamma_1)$.\\

Recall that $\gamma_1$ is positive semidefinite. Since $\Im(\gamma_2)\subset\Im(\gamma_1)$, there exists  $\lambda\in\mathbb{R}$ such that $\gamma_1-\lambda\gamma_2$ still is positive semidefinite with $0\neq v\in\Im(\gamma_1)\cap\ker(\gamma_1-\lambda\gamma_2)$, by lemma \ref{lemmaY}. Let $\delta_1=\gamma_1-\lambda\gamma_2/|\gamma_1-\lambda\gamma_2|$.

Let $\{\delta_1,\ldots,\delta_n\}$ be an orthonormal basis of $supp_1(A)$ containing $\delta_1$. So $A=\delta_1\otimes G_A(\delta_1)+\ldots+\delta_n\otimes G_A(\delta_n)=\delta_1\otimes\delta_1+\ldots+\delta_n\otimes\delta_n$.

Now if $\Im(\delta_i)\subset\Im(\delta_1)$ for every $i$, since $\ker(\delta_i)=\Im(\delta_i)^{\perp}$ (because $\delta_i$ are Hermitian matrices), we got $\ker(\delta_i)\supset\ker(\delta_1)$. So $tr(\delta_iv\overline{v}^t)=0$ for every $i$.

So $F_A(v\overline{v}^t)=\sum_{i=1}^n\delta_i tr(\delta_iv\overline{v}^t)=0$, but $F_A(v\overline{v}^t)=\sum_{i=1}^n\gamma_i tr(\gamma_iv\overline{v}^t)$.

Since $\gamma_i$ are orthonormal $tr(\gamma_iv\overline{v}^t)=0$ for every $i$, so $tr(\gamma_1v\overline{v}^t)=0$, but this is not possible since $\gamma_1$ is positive semidefinite and $0\neq v\in\Im(\gamma_1)$. Thus, exists $i>1$ such that $\Im(\delta_i)$ is not contained in $\Im(\delta_1)$.

By the first case we can write $A=A_1+A_2$, with the same type of $A$, but both with smaller tensor rank. Let us now use induction on the tensor rank.
 \\\\
This induction proves the existence of this Hermitian Schmidt decomposition for $A$.
Let us now prove the uniqueness of this decomposition. \\\\
b) Since $\{\gamma_1',...,\gamma_n'\}$ are positive semidefinite and orthonormal then $\Im(\gamma_i')\perp\Im(\gamma_j')$ for $i\neq j$. 
Now since $A= \sum_{i=1}^n\gamma_i'\otimes\gamma_i'$, we know that $$ \Im(A)=\bigoplus_{i=1}^n\Im(\gamma_i'\otimes\gamma_i')=\bigoplus_{i=1}^n\Im(\gamma_i')\otimes\Im(\gamma_i').$$

Suppose that $A$ has another Hermitian Schmidt decomposition, $A= \sum_{i=1}^n\delta_i\otimes\delta_i$, such that $\delta_1,...,\delta_n$ are positive semidefinite.
Again, we obtain $\Im(A)=\bigoplus_{i=1}^n\Im(\delta_i)\otimes\Im(\delta_i).$

Now $\delta_j=G_A(\delta_j)=\sum_{i=1}^ntr(\gamma_i'\delta_j)\gamma_i'$. 

Since $\gamma_i'$ and $\delta_j$ are positive semidefinite then $tr(\gamma_i'\delta_j)\geq0$. 

If there are two different indeces $r,s$ such that $tr(\gamma_r'\delta_j)>0$ and $tr(\gamma_s'\delta_j)>0$, we obtain
$\Im(\gamma_r')\oplus\Im(\gamma_s')\subset\Im(\delta_j).$ Thus, $\Im(\gamma_r')\otimes\Im(\gamma_s')\subset\Im(\delta_j)\otimes\Im(\delta_j)\subset\Im(A).$

But $\Im(\gamma_r')\otimes\Im(\gamma_s')\perp\Im(\gamma_i')\otimes\Im(\gamma_i')$, for every $1\leq i\leq n$, thus $$\Im(\gamma_r')\otimes\Im(\gamma_s')\perp\Im(A),$$  a contradiction

Therefore  $\delta_j=G_A(\delta_j)=tr(\gamma_r'\delta_j)\gamma_r'$ just for one index $r$. Now, since both $\delta_j$ and $\gamma_r'$ are positive semidefinite with norm equal to 1, then $\delta_j=\gamma_r'$.

Finally, each $\delta_j$ is equal to some $\gamma_r'$ and the decompositions are equal.
\end{proof} 
\vspace{0.5cm}

Our first theorem regarding PPT matrices is the following corollary.

\begin{corollary}\label{corollarytheomiracle1}
Let $A\in M_k\otimes M_m$ be a PPT matrix with Hermitian Schmidt decomposition $A=\sum_{i=1}^n\gamma_i\otimes\delta_i$ $($all $\lambda_i=1)$. Then $A$ has a unique Hermitian Schmidt decomposition,
$A= \sum_{i=1}^n\gamma_i'\otimes\delta_i'$,  such that $\gamma_i',\delta_i'$ are positive semidefinite Hermitian matrices for $1\leq i\leq n$.
\end{corollary}
\begin{proof}
Let $\sum_{i=1}^n\gamma_i\otimes\delta_i$ be a Hermitian Schmidt decomposition of $A$.

Since $A$ is PPT, the matrix $\sum_{i=1}^n\gamma_i^t\otimes\delta_i$ is positive semidefinite.

Therefore $(\sum_{i=1}^n\gamma_i^t\otimes\delta_i)^t=\sum_{i=1}^n\gamma_i\otimes\delta_i^t$ is positive semidefinite and $B=\sum_{i=1}^n\delta_i^t\otimes\gamma_i$ is too.

Now $A*B=\sum_{i=1}^n\gamma_i\otimes\gamma_i$ is positive semidefinite.

By theorem \ref{theomiracle1}, $A*B$ has another Hermitian Schmidt decomposition $A*B=\sum_{i=1}^n\gamma_i'\otimes\gamma_i'$ such that $\gamma_i'$ is positive semidefinite for $1\leq i\leq n$.

Now notice that $supp_1(A*B)=supp_1(A)$ and $\{\gamma_1',...,\gamma_n'\}$ is an orthonormal basis of $supp_1(A*B)=supp_1(A)$.

Therefore $A=\sum_{i=1}^n\gamma_i'\otimes G_A(\gamma_i')$. 

Now since $\gamma_i'\otimes\gamma_i'^t$ and $A$ are positive semidefinite and
$\gamma_i'\otimes\gamma_i'^t*A=\gamma_i'\otimes G_A(\gamma_i'),$
then the matrices $G_A(\gamma_1'),...,G_A(\gamma_n')$ are positive semidefinite.

The matrices $G_A(\gamma_1'),...,G_A(\gamma_n')$ are also orthonormal because the adjoints maps of the remark \ref{remarkadjoint} $($for this $A):$  $$F_A: supp_2(A)\rightarrow supp_1(A),\ \ G_A: supp_1(A)\rightarrow supp_2(A)$$ satisfy $F_A\circ G_A=Id$.
Thus $G_A$ is an isometry and since $\{\gamma_1',...,\gamma_n'\}$ are orthonormal the matrices $G_A(\gamma_1'),...,G_A(\gamma_n')$ are orthonormal.

Let $G_A(\gamma_i')=\delta_i'$ for $1\leq i\leq n$.

Finally the Hermitian Schmidt decomposition required for A, in this corollary, is
$A=\sum_{i=1}^n\gamma_i'\otimes \delta_i'$.\\

For the uniqueness of such decomposition, notice that  the decomposition  of $A*B$, $\sum_{i=1}^n\gamma_i'\otimes\gamma_i'$, is unique by item b of theorem \ref{theomiracle1}. Therefore $\{\gamma_1',...,\gamma_n'\}$ is an orthonormal basis of $supp_1(A*B)=supp_1(A)$ such that all matrices are positive semidefinite.

Suppose $\{\gamma_1'',...,\gamma_n''\}$ is  another orthonormal basis of $supp_1(A*B)=supp_1(A)$, such that all matrices are positive semidefinite, then $A*B=\sum_{i=1}^n\gamma_i''\otimes G_{A*B}(\gamma_i'')=\sum_{i=1}^n\gamma_i''\otimes\gamma_i''$ is another Hermitian Schmidt decomposition of $A*B$ such that all matrices are positive semidefinite, which is absurd.

Therefore $\{\gamma_1',...,\gamma_n'\}$ is the only orthonormal basis of $supp_1(A)$ such that all matrices are positive semidefinite and $A=\sum_{i=1}^n\gamma_i'\otimes G_A(\gamma_i')$ is the only Hermitian Schmidt decomposition announced in this corollary.
\end{proof}

\begin{lemma}\label{lemmapositive}  Let $A\in M_k\otimes M_m$ be a positive semidefinite Hermitian matrix. Let $\sum_{i=1}^n\lambda_i\gamma_i\otimes\delta_i$ be a Hermitian Schmidt decomposition of A such that  $\lambda_1=...=\lambda_s>\lambda_{s+1}\geq...\geq\lambda_n>0$. 

Then $D=\sum_{i=1}^s\gamma_i\otimes\delta_i$ is a positive semidefinite hermitan matrix.
\end{lemma}

\begin{proof}
The matrix $B=\sum_{i=1}^s\gamma_i\otimes\delta_i+\sum_{i=s+1}^n\frac{\lambda_i}{\lambda_1}\gamma_i\otimes\delta_i$ is positive semidefinite.
It is well known that 
$C=\sum_{i=1}^s\delta_i^t\otimes\gamma_i^t+\sum_{i=s+1}^n\frac{\lambda_i}{\lambda_1}\delta_i^t\otimes\gamma_i^t$ is also  a positive semidefinite Hermitian matrix.

Therefore $B*C=\sum_{i=1}^s\gamma_i\otimes\gamma_i^t+\sum_{i=s+1}^n(\frac{\lambda_i}{\lambda_1})^2\gamma_i\otimes\gamma_i^t$  is a positive semidefinite Hermitian matrix by corollary \ref{corollarypositive}.

The $*-product$ of $B*C$ by itself, $l$ times, remains positive semidefinite and since $0<\frac{\lambda_i}{\lambda_1}<1$, for $i>s$, the\\
$\displaystyle\lim_{l\rightarrow\infty}\overbrace{(B*C)*...*(B*C)}^{l\  times}=\lim_{l\rightarrow\infty}\sum_{i=1}^s\gamma_i\otimes\gamma_i^t+\sum_{i=s+1}^n(\frac{\lambda_i}{\lambda_1})^{2l}\gamma_i\otimes\gamma_i^t=\sum_{i=1}^s\gamma_i\otimes\gamma_i^t,$\\
still is positive semidefinite (the set of positive semidefinite Hermitian matrices is closed!).

Thus $\sum_{i=1}^s\gamma_i\otimes\gamma_i^t*B=\sum_{i=1}^s\gamma_i\otimes\delta_i=D$ is a positive semidefinite Hermitian matrix.
\end{proof}

\subsection{Split Decompositions}
\indent\\

\begin{theorem}\textbf{$($Split Decomposition for SPC matrices$)$}\label{splitdecompositionSPC}
Let $A\in M_k\otimes M_k$ be a SPC matrix. Suppose that $\lambda_1=...=\lambda_s>\lambda_{s+1}\geq...\geq\lambda_n>0.$
 Let $D=\sum_{i=1}^s\gamma_i\otimes\gamma_i$.

\begin{itemize}
\item[a)]  There exists a unique Hermitian Schmidt decomposition, $D=\sum_{i=1}^s\gamma_i'\otimes\gamma_i'$, such that $\gamma_i'$ are positive semidefinite for $1\leq i\leq s$.
\item[b)] Let $V_i\in M_k$ be the Hermitian projection onto $\Im(\gamma_i')$  $(1\leq i\leq s)$ and  $V_{s+1}\in M_k$ be the Hermitian projection onto $(\Im(\gamma_1')\oplus ...\oplus\Im(\gamma_{s}'))^{\perp}$.\\
 Then 
\begin{equation}\label{eqsplit1}
A=\sum_{i=1}^{s+1} (V_i\otimes V_i) A (V_i\otimes V_i)
\end{equation}
and $(V_i\otimes V_i) A (V_i\otimes V_i)$ are also SPC, for $1\leq i\leq s+1$.
\end{itemize}

We denote equation \eqref{eqsplit1} as \textbf{the split decomposition} of $A$. 
 
\end{theorem}
\begin{proof} a) By lemma \ref{lemmapositive}, $D=\sum_{i=1}^s \gamma_i\otimes\gamma_i$ is positive semidefinite. Thus, by theorem \ref{theomiracle1}, the result follows.

b) First, since $tr(\gamma_i'\gamma_j')=0$, for $i\neq j$, and $\gamma_i'$ is positive semidefinite for every $i$ then $\Im(V_i)=\Im(\gamma_i')\subset\ker(\gamma_j')$. Now $\ker(\gamma_j')=\Im(\gamma_j')^{\perp}=\ker(V_j)$, since $\gamma_j'$ is Hermitian. Thus, $\Im(V_i)\subset\ker(V_j)$ and $V_iV_j=V_jV_i=0$ for $i\neq j$.

Let $V_j'=Id-V_j$ and notice that $V_j'V_i=V_iV_j'=V_i$, for $i\neq j$, since $V_iV_j=V_jV_i=0$.

Now by lemma \ref{corollarynew5},
$A=(V_i\otimes V_i)A(V_i\otimes V_i)+(V_i'\otimes V_i')A(V_i'\otimes V_i')$, for $1\leq i \leq s$. Thus,
$A=(V_1\otimes V_1)A(V_1\otimes V_1)+(V_1'\otimes V_1')A(V_1'\otimes V_1')$. 

Next $(V_1'\otimes V_1')A(V_1'\otimes V_1')=(V_1'\otimes V_1')(V_2\otimes V_2)A(V_2\otimes V_2)(V_1'\otimes V_1')+(V_1'\otimes V_1')(V_2'\otimes V_2')A(V_2'\otimes V_2')(V_1'\otimes V_1')$.

Thus, $A=(V_1\otimes V_1)A(V_1\otimes V_1)+(V_2\otimes V_2)A(V_2\otimes V_2)+(V_1'\otimes V_1')(V_2'\otimes V_2')A(V_2'\otimes V_2')(V_1'\otimes V_1')$, since $V_1'V_2=V_2V_1'=V_2$.
 
We can repeat the argument to obtain $A=\sum_{i=1}^s (V_s\otimes V_s)A(V_s\otimes V_s)+(V_1'\ldots V_s'\otimes V_1'\ldots V_s')A(V_1'\ldots V_s'\otimes V_1'\ldots V_s').$

Notice that  $V_1'\ldots V_s'=Id-V_1-\ldots-V_s=V_{s+1}$, because $V_jV_i=0$ for $i\neq j$.

Finally, to see that $(V_i\otimes V_i) A (V_i\otimes V_i)$ is also SPC, notice that $G_A(X)=F_A(X)=\sum_{i=1}^n \lambda_i tr(\gamma_i X)\gamma_i$. It implies that $G_A:supp_1(A)\rightarrow supp_2(A)=supp_1(A)$ is a positive semidefinite self-adjoint linear transformation. Thus, $F_A\circ G_A=G_{A}^2$ and $G_A$ have the same eigenvectors which implies that all Hermitian Schmidt decompositions of $A$ have the same SPC type $($See theorem \ref{SVD} for more details$)$.

Next, since $V_iV_j=0$ for $i\neq j$, then $supp_1((V_i\otimes V_i) A (V_i\otimes V_i))\perp supp_1((V_j\otimes V_j) A (V_j\otimes V_j))$. Thus, the sum of the Hermitian Schmidt decompositions of $(V_i\otimes V_i) A (V_i\otimes V_i)$ for $1\leq i\leq s+1$ is a Hermitian Schmidt decomposition of $A$, which is SPC. Thus, each $(V_i\otimes V_i) A (V_i\otimes V_i)$ is SPC.

\end{proof}

\begin{theorem}\textbf{$($Split Decomposition for PPT matrices$)$}\label{splitdecompositionPPT}
 Let $A\in M_k\otimes M_m$ be a PPT matrix. Suppose $A= \sum_{i=1}^n\lambda_i\gamma_i\otimes\delta_i$ is a Hermitian Schmidt decomposition such that $\lambda_1=...=\lambda_s>\lambda_{s+1}\geq...\geq\lambda_n>0.$
 Let $D=\sum_{i=1}^s\gamma_i\otimes\delta_i$.
 
\begin{itemize}
\item[a)]  There exists a unique Hermitian Schmidt decomposition, $D=\sum_{i=1}^s\gamma_i'\otimes\delta_i'$, such that $\gamma_i',\delta_i'$ are positive semidefinite for $1\leq i\leq s$.
\item[b)] Let $V_i\in M_k$ be the Hermitian projection onto $\Im(\gamma_i')$  $(1\leq i\leq s)$ and  $V_{s+1}\in M_k$ be the Hermitian projection onto $(\Im(\gamma_1')\oplus ...\oplus\Im(\gamma_{s}'))^{\perp}$.\\
Let $W_i\in M_m$ be the Hermitian projection onto $\Im(\delta_i')$  $(1\leq i\leq s)$ and $W_{s+1}\in M_m$ be the Hermitian projection  onto $(\Im(\delta_1')\oplus ...\oplus\Im(\delta_{s}'))^{\perp}$.\\
 Then 
\begin{equation}\label{eqsplit}
A=\sum_{i=1}^{s+1} (V_i\otimes W_i) A (V_i\otimes W_i)
\end{equation}
and $(V_i\otimes W_i) A (V_i\otimes W_i)$ is PPT for $1\leq i\leq s+1$.
\end{itemize}

We denote  equation \eqref{eqsplit} as  \textbf{the split decomposition} of $A$.

 \end{theorem}

\begin{proof} a) By lemma \ref{lemmapositive}, $D=\sum_{i=1}^s \gamma_i\otimes\delta_i$ is positive semidefinite.
Now since $A$ is PPT then $A^{t_1}= \sum_{i=1}^n \lambda_i\gamma_i^t\otimes\delta_i$ is positive semidefinite. 
Then by lemma \ref{lemmapositive}, $\sum_{i=1}^s \gamma_i^t\otimes\delta_i=(\cdot)^t\otimes Id (D)$ is positive semidefinite too.

Therefore $D=\sum_{i=1}^s \gamma_i\otimes\delta_i$ is positive semidefinite and PPT and by corollary \ref{corollarytheomiracle1} the result follows.\\\\
b)
First, by lemma \ref{new3}, since $\gamma_i'$ and $\delta_i'$ are positive semidefinite, then
$(\Im(\gamma_i')\otimes\ker(\delta_i'))\oplus(\ker(\gamma_i')\otimes\Im(\delta_i'))\subset\ker(A)$
for every $1\leq i\leq s$.

Next since $tr(\gamma_i'\gamma_j')=tr(\delta_i'\delta_j')=0$ and $\gamma_i'$ and $\delta_i'$ are positive semidefinite then $\Im(\gamma_i')\subset \ker(\gamma_j')$ and $\Im(\delta_i')\subset \ker(\delta_j'),$  for $1\leq i\neq j\leq s$. 
Notice also that $(\Im(\gamma_1')\oplus ...\oplus\Im(\gamma_{s}'))^{\perp}\subset \ker(\gamma_j')$ and 
$(\Im(\delta_1')\oplus ...\oplus\Im(\delta_{s}'))^{\perp}\subset\ker(\delta_j')$, for $1\leq  j\leq s$.

Therefore, for $1\leq i\neq j\leq s+1$, we have $\Im(V_i)\otimes \Im(W_j)\subset \Im(\gamma_i')\otimes\ker(\delta_i')$ or $\Im(V_i)\otimes \Im(W_j)\subset\ker(\gamma_j')\otimes\Im(\delta_j')$
depending if $i<s+1$ or $j<s+1$.
Thus,  $A(V_i\otimes W_j)=0$ and $[A(V_i\otimes W_j)]^*=(V_i\otimes W_j)A=0$,  for $1\leq i\neq j\leq s+1$.

Now $\displaystyle \sum_{i=1}^{s+1}V_i=Id\in M_k$ and $\displaystyle \sum_{j=1}^{s+1}W_j=Id\in M_m$ and 
\begin{equation}\label{eq7}
A=\sum_{i,j,p,q=1}^{s+1} (V_i\otimes W_j)A(V_p\otimes W_q)=\sum_{i,j=1}^{s+1} (V_i\otimes W_i)A(V_j\otimes W_j).
\end{equation}

Since $A$ is PPT then $A^{t_1}$ is a positive semidefinite Hermitian matrix. Notice that $A^{t_1}=\lambda_1(\sum_{i=1}^s\gamma_i'^t\otimes\delta_i')+\sum_{i=s+1}^n\lambda_i\gamma_i^t\otimes\delta_i$ and $V_i^t\in M_k$ is the Hermitian projection onto $\Im(\gamma_i'^t)$  $(1\leq i\leq s)$, $V_{s+1}^t\in M_k$ is the Hermitian projection onto $(\Im(\gamma_1'^t)\oplus ...\oplus\Im(\gamma_{s}'^t))^{\perp}$.

By the same reason that we obtained $A(V_i\otimes W_j)=0$, we obtain now $A^{t_1}(V_i^t\otimes W_j)=0$, for $1\leq i\neq j\leq s+1$.

Finally, from  equation \ref{eq7}, we get $A^{t_1}=(\cdot)^t\otimes Id(A)=$ $$\sum_{i,j=1}^{s+1} (V_j^t\otimes W_i)A^{t_1}(V_i^t\otimes W_j)=\sum_{i=1}^{s+1} (V_i^t\otimes W_i)A^{t_1}(V_i^t\otimes W_i).$$
Thus, $A=(\cdot)^t\otimes Id(A^{t_1})=\sum_{i=1}^{s+1} (V_i\otimes W_i)A(V_i\otimes W_i)$.

Now since $A^{t_1}$ is positive semidefinite  then 
$(V_i^t\otimes W_i)A^{t_1}(V_i^t\otimes W_i)$ is also positive semidefinite, and thus, $(V_i\otimes W_i)A(V_i\otimes W_i)$ is PPT.  Therefore is PPT.

\end{proof}

\section{A Description of Weak Irreducible SPC or PPT Matrices}

The aim of this section is to continue to prove theorems about SPC/PPT matrices, but first we need the definition  
of weak irreducible matrix. This definition is a weaker version of the concept of irreducible state recently defined in \cite{chen2}. As we can see in the last line of table I of the paper just cited, the matrix $Id\otimes Id$ is not an irreducible state, but in our definition this matrix will be and every irreducible matrix in their sense is irreducible in our sense. Thus, the restriction we impose on a matrix to be irreducible is weaker than the restriction that the authors of \cite{chen2} imposed.

Two important theorems proved in this section are theorems \ref{conditionsirreducible} and \ref{conditionsirreducible2}. These theorems provide a description of all weak irreducible SPC/PPT matrices. This description is related to the format of their Hermitian Schmidt decompositions. 

A relevant fact about this weak irreducible property is that the authors of \cite{leinaas} reduced the positive definite case of the separability problem to a certain standard type of matrices, as described in our introduction. These positive definite matrices with the standard type are all weak irreducible and this fact can be noticed using corollary \ref{corollarylemmainequality}. Actually, these matrices satisfy the condition in theorem \ref{conditionsirreducible} to be weak irreducible.

We also proved in this section that the SPC/PPT matrices are weak irreducible or sum of weak irreducible matrices of the same type. The importance of this kind of theorem for the separability problem was noticed by the authors of \cite{chen1} in their Corollary 16. There they noticed that the separability problem can be reduced to the set of irreducible matrices. 

For the sake of completeness we shall also show in the next section that the separability problem can be reduced to the set of weak irreducible matrices and let us not forget that we have a complete description  of all weak irreducible SPC/PPT matrices (theorem \ref{conditionsirreducible2}).
This reduction and this description of weak irreducible matrices can be seen as a generalization  of the result obtained in \cite{leinaas}, for the positive definite case, described in the previous paragraphs and in the introduction. 
Following this idea we could, in section 5, provide sharp inequalities for separability of weak irreducible SPC/PPT matrices.
Recall that a necessary condition for separability of a matrix is to be PPT. Therefore we obtained quite general results.

In the last result of this section we proved that if the tensor rank or the rank of a matrix is big enough then the matrix must be weak irreducible.

\begin{definition}\label{definitionweakirreducible} Let  $V_i:\mathbb{C}^k\rightarrow\mathbb{C}^k$ and $W_i:\mathbb{C}^m\rightarrow\mathbb{C}^m$ be Hermitian projections, for $1\leq i\leq 2$,  such that $V_1V_2=0$, $W_1W_2=0$, $V_1+V_2=Id\in M_k$  and $W_1+W_2=Id\in M_m$.

Let $A\in M_k\otimes M_m$ be a positive semidefinite Hermitian matrix. We say that $A$ is  \textbf{weak irreducible} if 
the equality $A=\sum_{i=1}^{2} (V_i\otimes W_i) A (V_i\otimes W_i)$ holds then either $(V_1\otimes W_1) A (V_1\otimes W_1)=0$ or  $(V_2\otimes W_2) A (V_2\otimes W_2)=0$.
\end{definition}

\subsection{A Condition for Weak Irreducibility}
\indent\\

\begin{theorem}\label{conditionsirreducible}
 Let $A\in M_k\otimes M_m$ be positive semidefinite Hermitian matrix with a Hermitian Schmidt decomposition
 $A=\sum_{i=1}^n\lambda_i\gamma_i\otimes\delta_i$ with $\lambda_1\geq\ldots\geq\lambda_n>0$. Then $A$ is weak irreducible if 
\begin{enumerate}
\item $\lambda_1>\lambda_{2}\geq...\geq\lambda_n>0$.
\item $\Im(\gamma_i)\subset\Im(\gamma_1)$ for $1\leq i\leq n$.
\item $\Im(\delta_i)\subset\Im(\delta_1)$ for $1\leq i\leq n$.
\end{enumerate}  
\end{theorem}
\begin{proof}
Suppose $A$ satifies these three conditions above, recall that $\lambda_i^2$ are the non null eigenvalues of the self-adjoint map $F_A\circ G_A: supp_1(A)\rightarrow supp_2(A)$ and $\gamma_i$ are corresponding eigenvectors.

Let $A=A_1+A_2$ where $A_i=(V_i\otimes W_i) A (V_i\otimes W_i)$ and $V_i,W_i$ are described in definition \ref{definitionweakirreducible}.

Notice that $supp_i(A_1)\perp supp_i(A_2)$, for $1\leq i\leq 2$, since $V_1V_2=0$, $W_1W_2=0$. 
Then  $supp_i(A)= supp_i(A_1)\oplus supp_i(A_2)$, for $1\leq i\leq 2$.

Consider the maps $F_{B}:supp_2(A)\rightarrow supp_1(A)$ and $G_{B}:supp_1(A)\rightarrow supp_2(A)$, for $B\in\{A,A_1,A_2\}$. Notice that we are using different domain and codomain when $B=A_1$ and $A_2$ (see definition \ref{defmaps} and remark \ref{remarkadjoint}).

Notice that $F_{A_i}\circ G_{A_j}=0$, for $i\neq j$, since $supp_2(A_1)\perp supp_2(A_2)$ and $\Im(G_{A_1})=supp_2(A_1)\subset\ker(F_{A_2})$,  $\Im(G_{A_2})=supp_2(A_2)\subset\ker(F_{A_1})$.

Thus, $F_A\circ G_A=(F_{A_1}+F_{A_2})\circ (G_{A_1}+G_{A_2})=F_{A_1}\circ G_{A_1}+F_{A_2}\circ G_{A_2}$.

Now $F_{A_i}\circ G_{A_i}(supp_1(A_i))\subset supp_1(A_i)$ and $F_{A_i}\circ G_{A_i}(supp_1(A_j))=0$, for $1\leq i\leq 2$ and $i\neq j$. Therefore, the non null eigenvalues  of $F_A\circ G_A$ are the non null eigenvalues  of $F_{A_i}\circ G_{A_i}$,  for $1\leq i\leq 2$. Since the multiplicity of $\lambda_1^2$ is 1 then  $\gamma_1$ is an eigenvector of $F_{A_1}\circ G_{A_1}$ or an eigenvector of $F_{A_2}\circ G_{A_2}$.

Suppose $\gamma_1$ is an eigenvector of $F_{A_1}\circ G_{A_1}$ and since $\lambda_1^2\neq 0$  then $\gamma_1\in supp_1(A_1)$. Notice that $supp_1(A_1)\subset  span\ \{V_1\gamma_1V_1,\ldots,V_1\gamma_nV_1\}$, by item b) of theorem \ref{theoremprincipal}, and since $\gamma_1\in supp_1(A_1)$, we obtain $\Im(\gamma_1)\subset\Im(V_1)$.  Recall that $\Im(V_1)\perp\Im(V_2)$.

Finally by our assumption, for every $1\leq i\leq n$, we know that $\Im(\gamma_i)\subset\Im(\gamma_1)\subset\Im(V_1)$. Thus, $V_2\gamma_iV_2=0$ and $(V_2\otimes W_2) A (V_2\otimes W_2)=0$.
Therefore $A$  is weak irreducible.
\end{proof}
\begin{remark} The converse of this theorem is false. For example, let $v=\lambda_1 e_1\otimes e_1+\lambda_2 e_2\otimes e_2\in \mathbb{C}^k\otimes\mathbb{C}^k$, where $e_1$ and $e_2$ are the first two vectors in the canonical basis of $\mathbb{C}^k$ and $\lambda_1>\lambda_2>0$.
Let $A=v\overline{v}^t$. Since $A$ has rank 1 then $A$ is obviously weak irreducible. Now a Hermitian Schmidt decomposition of $A$ is $\lambda_1^2 \gamma_1\otimes \gamma_1 +(\lambda_1\lambda_2) \gamma_2\otimes\gamma_2+(\lambda_1\lambda_2) \gamma_3\otimes\gamma_3^t+\lambda_2^2  \gamma_4\otimes \gamma_4,$ where 
$\gamma_1= e_1e_1^t, \gamma_2=\frac{e_1e_2^t+e_2e_1^t}{\sqrt{2}}, \gamma_3=\frac{i(e_1e_2^t-e_2e_1^t)}{\sqrt{2}}, \gamma_4=e_2e_2^t$.
Notice that $\lambda_1^2>\lambda_1\lambda_2>\lambda_2^2$, but $\Im(\gamma_i)$ is not contained in  $\Im(\gamma_1)$. 
\end{remark}

Thus, the converse of  theorem \ref{conditionsirreducible} is not true in general, but if the matrix is SPC or PPT then the converse is true.

\subsection{A Description of Weak Irreducible SPC/PPT Matrices}
\indent\\

\begin{theorem}\label{conditionsirreducible2}$($\textbf{Description of Weak Irreducible SPC$/$PPT Matrices}$)$.
 Let $A\in M_k\otimes M_m$ be a SPC or PPT matrix. Then $A$ is weak irreducible if and only if the Hermitian Schmidt decomposition of $A$ satisfies the three conditions of theorem \ref{conditionsirreducible}.
 \end{theorem}
\begin{proof}  If $A$ satisfies the three conditions, we saw that $A$ is weak irreducible in theorem \ref{conditionsirreducible}.

Now suppose $A$ is weak irreducible PPT matrix, let us prove that $A$ satifies the three conditions of theorem \ref{conditionsirreducible}.

Suppose $\lambda_1=\ldots=\lambda_{s}$ and $s>1$. Consider the notation of theorem \ref{splitdecompositionPPT}. 
By item b) of theorem \ref{splitdecompositionPPT}, $A=\sum_{i=1}^{s+1} (V_i\otimes W_i) A (V_i\otimes W_i)$. Let $V_2'=Id-V_1$ and $W_2'=Id-W_1$. Notice that $(V_2'\otimes W_2')A(V_2'\otimes W_2')=\sum_{i=2}^{s+1} (V_i\otimes W_i) A (V_i\otimes W_i)$, since $V_2'V_i=V_iV_2'=V_i$ and $W_2'W_i=W_iW_2'=W_i$ for $2\leq i\leq n$. Then $A=(V_1\otimes W_1)A(V_1\otimes W_1)+(V_2'\otimes W_2')A(V_2'\otimes W_2')$.

Now if $(V_1\otimes W_1)A(V_1\otimes W_1)=0$ then $supp_1(A)=supp_1((V_2'\otimes W_2')A(V_2'\otimes W_2'))$ then $\gamma_1'\in supp_1((V_2'\otimes W_2')A(V_2'\otimes W_2'))$. Then $\Im(\gamma_1')\subset\Im(V_2')\perp\Im(V_1)=\Im(\gamma_1').$ Thus, $\gamma_1'=0$ which is absurd.

Next since $A$ is weak irreducible we must have $(V_2'\otimes W_2')A(V_2'\otimes W_2')=0$, which implies $supp_1(A)=supp_1((V_1\otimes W_1)A(V_1\otimes W_1))$, but $\gamma_2'\in supp_1((V_1\otimes W_1)A(V_1\otimes W_1))$. Then $\Im(\gamma_2')\subset\Im(V_1)$ but $\Im(\gamma_2')=\Im(V_2)\perp\Im(V_1)$. Then, $\gamma_2'=0$ which is absurd, thus $s=1$. Therefore $\lambda_1>\lambda_2$ and  $D=\gamma_1\otimes\delta_1=\gamma_1'\otimes\delta_1'$.

Again, by item b) of theorem \ref{splitdecompositionPPT}, but now with $\lambda_1>\lambda_2$, we have $A=\sum_{i=1}^{2} (V_i\otimes W_i) A (V_i\otimes W_i)$, where $V_2$ is the Hermitian projection onto $\Im(\gamma_1')^{\perp}$ and $W_2$ is the Hermitian projection onto $\Im(\delta_1')^{\perp}$.
Again, by the same reasoning as above we obtain $A=(V_1\otimes W_1)A(V_1\otimes W_1)$.

Thus, for every $\gamma_i\in supp_1(A)$  and $\delta_i\in supp_2(A)$ we have  $\Im(\gamma_i)\subset\Im(V_1)=\Im(\gamma_1')=\Im(\gamma_1)$ and  $\Im(\delta_i)\subset\Im(W_1)=\Im(\delta_1')=\Im(\delta_1)$. Thus, $A$ satisfies the three conditions of theorem \ref{conditionsirreducible}.

Now the proof for SPC matrices is the same, we just need to use theorem \ref{splitdecompositionSPC} instead of \ref{splitdecompositionPPT} and that $\gamma_i=\delta_i$, $\gamma_i'=\delta_i'$  and $V_i=W_i$.
\end{proof}

\begin{remark}\label{remarkconditions} In the proof of the previous theorem we saw that if $A$ is weak irreducible then the split decomposition of $A$ is  $A=(V_1\otimes W_1)A(V_1\otimes W_1)$. Actually, we also saw the converse, because if $A=(V_1\otimes W_1)A(V_1\otimes W_1)$ then we saw that $\lambda_1>\lambda_2$ and for every $\gamma_i\in supp_1(A)$  and $\delta_i\in supp_2(A)$ we have  $\Im(\gamma_i)\subset\Im(\gamma_1)$ and  $\Im(\delta_i)\subset\Im(\delta_1)$. Thus, $A$ satisfies the three conditions of theorem \ref{conditionsirreducible} and $A$ is weak irreducible. Thus, $A$ is weak irreducible if and only if its split decomposition is $A=(V_1\otimes W_1)A(V_1\otimes W_1)$.
\end{remark}

\begin{theorem}\label{irreducibleparts} Let $A$ be the SPC matrix of theorem \ref{splitdecompositionSPC}. The matrices $(V_j\otimes V_j)A(V_j\otimes V_j)$ of the split decomposition of $A$ $($equation \eqref{eqsplit1}$)$ for $1\leq j\leq s$ are weak irreducible.
Let $A$ be the PPT  matrix of theorem \ref{splitdecompositionPPT}. The matrices $(V_j\otimes W_j)A(V_j\otimes W_j)$ of the split decomposition of $A$ $($equation \eqref{eqsplit}$)$ for $1\leq j\leq s$ are weak irreducible.
\end{theorem}
\begin{proof} Let us prove this theorem only for PPT matrices. The proof for SPC matrices is similar.

Using the same notation as in theorem \ref{splitdecompositionPPT}, $A$ has the following Hermitian Schmidt decomposition:  $A=\lambda_1(\sum_{i=1}^s\gamma_i'\otimes\delta_i')+\sum_{i=s+1}^n\lambda_i\gamma_i\otimes\delta_i$
such that $\Im(\gamma_i')\perp \Im(\gamma_j')$ and $\Im(\delta_i')\perp\Im(\delta_j')$.

Since $V_j$ and $W_j$ are the Hermitian projections onto $\Im(\gamma_j')$ and $\Im(\delta_j')$ respectively then $V_j\gamma_j'V_j=\gamma_j'$, $V_j\gamma_i'V_j=0$, $W_j\delta_j'W_j=\delta_j'$ and $W_j\delta_i'W_j=0$. Thus, $(V_j\otimes W_j)A(V_j\otimes W_j)=\lambda_1\gamma_j'\otimes\delta_j'+\sum_{i=s+1}^n\lambda_i V_j\gamma_iV_j\otimes W_j\delta_iW_j$. 

Notice that $tr(\gamma_j' V_j\gamma_i V_j)=tr(V_j\gamma_j' V_j\gamma_i)=tr(\gamma_j' \gamma_i)=0$ and $tr(\delta_j' W_j\delta_i W_j)=tr(W_j\delta_j' W_j\delta_i)=tr(\delta_j' \delta_i)=0$ for $s+1\leq i\leq n$.

Let $\sum_{p=1}^q\epsilon_p \alpha_p\otimes\beta_p$ be an Hermitian Schmidt decomposition for $B=\sum_{i=s+1}^n\lambda_i V_j\gamma_iV_j\otimes W_j\delta_iW_j$ such that $\epsilon_1\geq\ldots\geq\epsilon_q>0$. By theorem \ref{theoremprincipal}, $supp_1(B)\subset\text{span}\{V_j\gamma_iV_j, s+1\leq i\leq n\}$ and $supp_2(A)\subset\text{span}\{W_j\delta_iW_j, s+1\leq j\leq n\}$. Then $supp_1(B)\perp\{\gamma_j'\}$ and $supp_2(B)\perp\{\delta_j'\}$.

Thus, $\lambda_1\gamma_j'\otimes\delta_j'+\sum_{p=1}^q\epsilon_p \alpha_p\otimes\beta_p$ is a Hermitian Schmidt decomposition for $(V_j\otimes W_j)A(V_j\otimes W_j)$.

Next, since $supp_t((V_j\otimes W_j)A(V_j\otimes W_j)) \perp supp_t((V_i\otimes W_i)A(V_i\otimes W_i))$ for $t=1,2$ and $1\leq i\neq j\leq s+1$, then the sum of the Hermitian Schmidt decompositions of each $(V_j\otimes W_j)A(V_j\otimes W_j)$
is a Hermitian Schmidt decomposition for $A$.

Recall that in the Hermitian Schmidt decomposition of $A$, we have  $\lambda_1=\ldots=\lambda_s>\lambda_{s+1}\geq\ldots\geq\lambda_n$. 
Then $\epsilon_1\leq\lambda_1$. 

Next, each $(V_i\otimes W_i)A(V_i\otimes W_i)$ has a Hermitian Schmidt decomposition with the term $\lambda_1\gamma_i'\otimes\delta_i'$, for $1\leq i\leq s$.
Now if we have $\epsilon_1=\lambda_1$ then $A$ would have a Hermitian Schmidt decomposition with more than $s$ numbers $\lambda_1$, but this is impossible.  Therefore $\lambda_1>\epsilon_1\geq\ldots\geq\epsilon_q>0$.

Finally, as we saw above $supp_1(B)\subset\text{span}\{V_j\gamma_iV_j, s+1\leq i\leq n\}$ and $supp_2(B)\subset\text{span}\{W_j\delta_iW_j, s+1\leq j\leq n\}$. Since $\alpha_p\in supp_1(B)$ and $\beta_p\in supp_2(B)$ ($1\leq p\leq q$)then $\Im(\alpha_p)\subset\Im(V_j)=\Im(\gamma_j')$ and $\Im(\beta_p)\subset\Im(W_j)=\Im(\delta_j')$.

Thus, by theorem \ref{conditionsirreducible},  $(V_j\otimes W_j)A(V_j\otimes W_j)$ is weak irreducible.
\end{proof}

\begin{corollary}\label{theoremsumofweakirreducible} Every SPC or PPT matrix $A$ is weak irreducible or a sum of weak irreducible matrices of the same type.
\end{corollary}
\begin{proof}
Consider a PPT matrix $A$.
By the last theorem we only need to deal with  $(V_{s+1}\otimes W_{s+1})A(V_{s+1}\otimes W_{s+1})$ in the split decomposition of a PPT matrix $A$ $($equation \ref{eqsplit}$)$.

If $A$ has tensor rank 1 then $A$ satisfies trivially the three conditions of theorem \ref{conditionsirreducible}, therefore $A$ is weak irreducible.

Suppose $A$ has tensor rank bigger than 1 and consider its split decomposition $($equation \eqref{eqsplit}$)$.

Notice that $supp_t((V_{j}\otimes W_{j})A(V_{j}\otimes W_{j}))\perp supp_t((V_{i}\otimes W_{i})A(V_{i}\otimes W_{i}))$ for $t=1,2$ and for every $1\leq i\neq j\leq s+1$ since $V_jV_i=0$ and $W_jW_i=0$. Thus, the tensor rank of $A$ is the sum of the tensor rank of each of these $(V_{i}\otimes W_{i})A(V_{i}\otimes W_{i})$,$1\leq i\leq s+1$. Therefore  $(V_{s+1}\otimes W_{s+1})A(V_{s+1}\otimes W_{s+1})$ has tensor rank smaller than $A$, because  $(V_{1}\otimes W_{1})A(V_{1}\otimes W_{1})\neq 0$. Recall that $(V_{s+1}\otimes W_{s+1})A(V_{s+1}\otimes W_{s+1})$ is also PPT as explained in theorem \ref{splitdecompositionPPT}. Thus, by induction on the tensor rank, the matrix $(V_{s+1}\otimes W_{s+1})A(V_{s+1}\otimes W_{s+1})$ is weak irreducible or a sum of weak irreducible PPT matrices.

The proof is similar for SPC matrices. 
\end{proof}

Now in the final result of this section we prove that if a matrix  has full tensor rank or full rank then the matrix is weak irreducible.

\begin{lemma}\label{lemmainequality} Let $k_i,k,m_i,m\in\mathbb{N}$, for $1\leq i\leq 2$, such that $k_1+k_2=k$ and $m_1+m_2=m$  then 
\begin{enumerate}
\item $k_1^2+k_2^2\leq (k-1)^2+1$
\item $k_1m_1+k_2m_2\leq \min\{(k-1)m,(m-1)k\}.$

\end{enumerate}

\end{lemma}
\begin{proof}
Since $k_1+k_2=k$ and $k_1,k_2\in\mathbb{N}$, we have $1\leq k_i\leq k-1$.

Notice that  $k_1^2+k_2^2=(k-k_2)^2+k_2^2=2k_2^2-2kk_2+k^2$.

Consider $p(x)=2x^2-2kx+k^2$, the line of simetry is $x=\frac{k}{2}$. If $x\in\mathbb{N}$ and $1 \leq x\leq k-1$ then the maximum of $p(x)$ occurs in $x=1$ or $k-1$.

Therefore $k_1^2+k_ 2^2\leq p(k-1)=(k-1)^2+1. $

Finally, since $k_i\geq 1$ and $m_i\geq 1$ then $k_1m_2+k_2m_1\geq\max\{k,m\}$.
Then $k_1m_1+k_2m_2=km-(k_1m_2+k_2m_1)\leq\min\{km-k,km-m\}$.

\end{proof}

\begin{theorem}\label{corollarylemmainequality} Let $A\in M_k\otimes M_m$ be a positive semidefinite Hermitian matrix which is not weak irreducible then
\begin{enumerate}
\item tensor rank $(A)\leq\min\{(k-1)^2+1,(m-1)^2+1\}$
\item rank $(A)\leq\min\{(k-1)m,(m-1)k\}.$
\end{enumerate}
\end{theorem}
\begin{proof}
Let $A=A_1+A_2$ where $A_i=(V_i\otimes W_i) A (V_i\otimes W_i)$ and $V_i,W_i$ are described in definition \ref{definitionweakirreducible} and $A_1,A_2\neq 0$. 

Let $k_i$ and $m_i$ be the rank of $V_i$ and $W_i$, respectively, for $1\leq i\leq 2$, and recall that $k_1+k_2=k$ and $m_1+m_2=m$.
Notice that $A_i$ can be embedded in $M_{k_i}\otimes M_{m_i}$ then the tensor rank of $A_i$ is smaller or equal to $\min\{k_i^2,m_i^2\}$ and $rank(A_i)\leq k_im_i$.

As pointed out in the proof of  theorem \ref{theoremsumofweakirreducible}, we have $tensor\ rank(A)= tensor\ rank(A_1)+tensor\ rank(A_2)$.

Thus, $tensor\ rank(A)\leq \sum_{i=1}^2\min\{k_i^2,m_i^2\}\leq \min\{\sum_{i=1}^2k_i^2,\sum_{i=1}^2m_i^2\}$.

Since $\Im(V_i\otimes W_i)\perp\Im(V_j\otimes W_j)$, we obtain $rank(A)= rank(A_1)+rank(A_2)\leq\sum_{i=1}^2k_im_i.$

By lemma \ref{lemmainequality}, the result follows.

\end{proof}

\begin{corollary} Let $A\in M_k\otimes M_m$ be a positive semidefinite Hermitan matrix. If $A$ has full tensor rank or full rank then $A$ is weak irreducible.
\end{corollary}

\section{The Separability Problem}

The separability problem, in finite dimension, is the problem of distinguishing the separable matrices among the positive semidefinite Hermitian matrices or, equivalently, it is the problem of discovering whether a given matrix is separable or not. It is a central problem in the field of Quantum Information Theory. 

We start this section defining the separable matrices and for the sake of completeness we proved that it is sufficient to distinguish the separable matrices among the weak irreducible matrices in order to solve completely the separability problem.
Actually, we proved that to solve the separability problem for SPC or PPT matrices, we must do it only for weak irreducible matrices of the same type. Just notice that a necessary condition for separability is to be PPT (see proposition \ref{reductionfinal} and corollary \ref{corollaryreductionfinal}). The authors of \cite{chen1} have already noticed that the reduction to their irreducible matrices was possible (see \cite{chen1}, corollary 16). 

The main result of this section is that every positive semidefinite Hermitian matrix in $M_k\otimes M_m$ with tensor rank 2 has a minimal separable decomposition in $M_k\otimes M_m$ $($Theorem \ref{theoreminimalseparabilitytensorrank2}$)$. Therefore is separable.
The separable decomposition of a tensor rank 2 matrix might be a known fact since it is simpler to be proved, but we believe that the minimality of the separable decomposition might be new, since the only proof we know is based on the concept of weak irreducible matrix and this concept was defined based on the concept of irreducible states, which was only defined recently in \cite{chen2}.
 
Notice that the minimality is something very important to extend this theorem to $M_{k_1}\otimes\ldots \otimes M_{k_n}$ $($Corollary \ref{extension}$)$.

We will need theorem \ref{theoreminimalseparabilitytensorrank2} in our final section. There we find sharp inequalities that provide separability for weak irreducible SPC/PPT matrices.

\begin{definition}\label{definitionseparability}\textbf{$($Separable Matrices$)$} Let $A\in M_k\otimes M_m$. 
\begin{enumerate}
\item A separable decomposition  of a Hermitian matrix A is a decomposition $A=\sum_{i=1}^n C_i\otimes D_i$ such that   $C_i\in M_k$ and $D_i\in M_m$ are positive semi-definite Hermitian matrices for every $i$. 
\item $A$ is \textbf{separable} if $A$ has a separable decomposition.
\end{enumerate}
\end{definition}

\begin{definition}\textbf{(The Separability Problem)}: Which matrices in $M_k\otimes M_m$ are separable?  Provide a criterion to determine if $A\in M_k\otimes M_m$ is separable or not.
\end{definition}

\subsection{The Reduction to the Weak Irreducible case}
\indent\\

\begin{proposition}\label{reductionfinal} To distinguish the separable matrices among the SPC/PPT matrices, we only need to do it among the matrices which are also weak irreducible. 
\end{proposition}
\begin{proof} 

Let $A$ be a PPT matrix. The proof for SPC matrices is similar.

If $A$ is weak irreducible then we must solve the separability problem for a weak irreducible PPT matrix. 

Suppose $A$ is not weak irreducible. By remark \ref{remarkconditions}, the split decomposition of $A$ has more then one term. We saw in theorem \ref{splitdecompositionPPT} that each term is also PPT, but these terms have smaller tensor rank because their supports are perpendicular. It is obvious that a matrix $A$ is separable if and only if each term in its split decomposition is separable (see theorem \ref{splitdecompositionPPT}). Then the result follows by induction on the tensor rank, just notice that matrices with tensor rank 1 satisfy trivially the three conditions to be weak irreducible.

\end{proof}

\begin{corollary}\label{corollaryreductionfinal}
A complete solution for the separability problem is obtained distinguishing the separable matrices among the weak irreducible matrices.
\end{corollary}
\begin{proof}
Notice that the PPT property is necessary for the separability of $A$. Therefore we need to distinguish the separable matrices among the PPT matrices. Use now the previous proposition. 
\end{proof}

\subsection{Minimal Separability of Tensor Rank 2 Matrices}
\indent\\

Now to obtain our sharp inequalities that provide separability in the next section, we need some theorems concerning the separability of positive semidefinite Hermitian matrices with tensor rank 2 (in $M_k\otimes M_m$). We provide these theorems as the final part of this section.

\begin{lemma}\label{lemmaseparable} Let $A\in M_{k}\otimes M_{ m}$ be a positive semidefinite Hermitian matrix with a minimal hermitan decomposition  $A=C_1\otimes D_1+C_2\otimes D_2$ such that $C_1,D_1$ are positive semidefinite and $\Im(C_2)\subset\Im(C_1), \Im(D_2)\subset\Im(D_1)$ then $A$ has a minimal separable decomposition. 
\end{lemma}   
\begin{proof}  

Choose $\lambda\in\mathbb{R}$ such that $C_1-\lambda C_2$ is positive semidefinite, with $0\neq v\in\ker(C_1-\lambda C_2)\cap\Im(C_1)$ (lemma \ref{lemmaY}). Notice that $\lambda\neq 0$.

Therefore
$A=(C_1-\lambda C_2)\otimes D_1+ C_2\otimes(D_2+\lambda D_1)$.

Since  $tr((C_1-\lambda C_2)v\overline{v}^t)=0$, we obtain $Id\otimes\overline{v}v^t*A=Id\otimes tr(C_2v\overline{v}^t)(D_2+\lambda D_1)$.
Thus, by corollary \ref{corollarypositive}, $tr(C_2v\overline{v}^t)(D_2+\lambda D_1)$ is positive semidefinite. 

Notice  that $0\neq\frac{tr(C_1v\overline{v}^t)}{\lambda}=tr(C_2v\overline{v}^t)$, since $v\in\Im(C_1)$ and $C_1$ is positive semidefinite.
Now let
 
$\beta_1=D_1$,

$\beta_2=tr(C_2v\overline{v}^t)(D_2+\lambda D_1),$

 $\alpha_1= C_1-\lambda C_2$,
 
  $\alpha_2=\frac{C_2}{tr(C_2v\overline{v}^t)}$.\\

Notice that $\alpha_1,\beta_2$ are positive semidefinite and $\beta_1$ is positive semidefinite such that $\Im(\beta_2)\subset\Im(\beta_1)$ and 
 $A=\alpha_1\otimes\beta_1+\alpha_2\otimes\beta_2.$
 
Now find a real number $\epsilon$ such that $\beta_1-\epsilon\beta_2$ is  positive semidefinite and has $0\neq w\in\ker(\beta_1-\epsilon\beta_2)\cap\Im(\beta_1)$ (lemma \ref{lemmaY}). Notice that $\epsilon\neq 0$.

Therefore
$$A=\alpha_1\otimes(\beta_1-\epsilon\beta_2)+(\alpha_2+\epsilon\alpha_1)\otimes\beta_2.$$ 

Since $tr((\beta_1-\epsilon\beta_2)w\overline{w}^t)=0$ then $A*\overline{w}w^t\otimes Id=tr(\beta_2w\overline{w}^t)(\alpha_2+\epsilon\alpha_1)\otimes Id$.
Thus, by corollary \ref{corollarypositive}, $tr(\beta_2w\overline{w}^t)(\alpha_2+\epsilon\alpha_1)$ is positive semidefinite.
 
Note also  that $0\neq\frac{tr(\beta_1w\overline{w}^t)}{\epsilon}=tr(\beta_2w\overline{w}^t)$, since $\beta_1$ is positive semidefinite and $w\in\Im(\beta_1)$.

 Since $tr(\beta_2w\overline{w}^t)>0$, by the positive semidefiniteness of $\beta_2$, we obtain the following minimal separable decomposition for A:
 
$$A=\alpha_1\otimes(\beta_1-\epsilon\beta_2)+tr(\beta_2w\overline{w}^t)(\alpha_2+\epsilon\alpha_1)\otimes\frac{\beta_2}{tr(\beta_2w\overline{w}^t)}.$$ 
\end{proof}

\begin{theorem} \label{theoremtensorrank2PPT}
If $A\in M_k\otimes M_m$ is a positive semidefinite Hermitian matrix with tensor rank 2 $($see definition \ref{definitiontensorrank}$)$ then $A$ is PPT. 
\end{theorem}
\begin{proof}

Let $A_1\otimes B_1+ A_2\otimes B_2$ be a minimal Hermitian decomposition of $A$. Since $A*Id\otimes Id=tr(B_1)A_1+tr(B_2)A_2\otimes Id$ is positive semidefinite, by lemma \ref{corollarypositive}, then $tr(B_1)A_1+tr(B_2)A_2$ is positive semidefinite. 

Now if $tr(B_1)A_1+tr(B_2)A_2=0$, since $A_1,A_2$ are linear independent  then $tr(B_1)=tr(B_2)=0$. Thus  $tr(A)=tr(A_1)tr(B_1)+tr(A_2)tr(B_2)=0$. Since $A$ is positive semidefinite it implies that $A=0$, which is a contradiction.
Thus, $C_1=\frac{tr(B_1)A_1+tr(B_2)A_2}{|tr(B_1)A_1+tr(B_2)A_2|}\neq 0$ is positive semidefinite.

Let $\{C_1,C_2\}$ be an orthonormal basis of $supp_1(A)$ containing $C_1$.
Then $A=C_1\otimes D_1+C_2\otimes D_2$. Since $Id\otimes C_1^t*A=Id\otimes D_1$ is positive semidefinite, by corolary \ref{corollarypositive}, then $D_1$ is positive semidefinite.

Thus, $A=C_1\otimes D_1+C_2\otimes D_2$ is a Hermitian Schmidt decomposition such that $C_1$ and $D_1$ are positive semidefinite. 

Let $\epsilon>0$ and define $A(\epsilon)=(C_1+\epsilon Id)\otimes (D_1+\epsilon Id)+C_2\otimes D_2$.
Notice that $A(\epsilon)$ has tensor rank at most two. If the tensor rank is one then $A(\epsilon)$ is separable, therefore PPT. If the tensor rank is two then $A(\epsilon)$ has  a minimal Hermitian decomposition with $\Im(C_2)\subset\Im(C_1+\epsilon Id) $ and $\Im(D_2)\subset\Im(D_1+\epsilon Id) $, thus, by lemma \ref{lemmaseparable}, $A(\epsilon)$ is also separable and PPT.

Thus, $(\cdot)^t\otimes Id(A)=\lim_{\epsilon\rightarrow 0+}(\cdot)^t\otimes Id(A(\epsilon))$ is positive semidefinite.

Therefore  $A$ is PPT. 
\end{proof}

\begin{theorem}\label{theoreminimalseparabilitytensorrank2} $($\textbf{Minimal Separability of Tensor Rank 2 Matrices}$)$
If $A\in M_k\otimes M_m$ is a positive semidefinite Hermitian matrix with tensor rank 2 then $A$ has a minimal separable decomposition. Therefore $A$ is separable.
\end{theorem}

 \begin{proof}
 Let 
 $A=\lambda_1\gamma_1\otimes\delta_1+ \lambda_2\gamma_2\otimes\delta_2$ be an Hermitian Schmidt decomposition of $A$ with $\lambda_1\geq\lambda_2>0$.
 
 First suppose $\lambda_1=\lambda_2$. By theorem \ref{theoremtensorrank2PPT}, $A$ is PPT then $\frac{1}{\lambda_1}A$ has  a minimal separable decomposition by theorem \ref{corollarytheomiracle1}. 
 
 Now let us suppose $\lambda_1>\lambda_2$.

 Since $A$ is PPT, by item a) of  theorem \ref{splitdecompositionPPT}, we  know that $\gamma_1\otimes\delta_1=\gamma_1'\otimes\delta_1'$ such that $\gamma_1',\delta_1'$ are positive semidefinite. 
 
 If $A$ is weak  irreducible then $\Im(\gamma_2)\subset\Im(\gamma_1')$ and $\Im(\delta_2)\subset\Im(\delta_1')$, by theorem \ref{conditionsirreducible2} and by lemma \ref{lemmaseparable} $A$ has a minimal separable decomposition in $M_k\otimes M_m$.

Now suppose that $A$ is not weak irreducible. 
Then $A=A_1+A_2$ where $A_i=(V_i\otimes W_i) A (V_i\otimes W_i)$ and $V_i,W_i$ are described in definition \ref{definitionweakirreducible} and $A_1,A_2\neq 0$.

As pointed out in the proof of theorem \ref{theoremsumofweakirreducible}, we have $tensor\ rank(A)=tensor\ rank(A_1)+\ tensor\  rank(A_2)$. Therefore $2= tensor\ rank(A_1)+tensor\ rank(A_2)$ and $tensor\ rank(A_i)=1$.

However, both $A_1,A_2$ are positive semidefinite with tensor rank 1. So $A$ has a minimal separable decomposition.
 
 \end{proof}

In the following corollary we can see how important the minimal separable decomposition of a positive semidefinite Hermitian matrix in $M_k\otimes M_m$ is, with tensor rank 2, to extend the same result for $M_{k_1}\otimes\ldots\otimes M_{k_n}$.

\begin{corollary} \label{extension}If $A\in M_{k_1}\otimes\ldots\otimes M_{k_n}$ is a positive semidefinite Hermitian matrix with tensor rank smaller or equal to 2 then $A$ is separable.
\end{corollary}
\begin{proof} 
If $A$ has tensor rank smaller or equal to 2 in $M_{k_1}\otimes\ldots\otimes M_{k_n}$, we can write $A=A_1\otimes A_2\otimes\ldots\otimes A_n+B_1\otimes B_2\otimes \ldots\otimes B_n.$  Thus, $A$ has tensor rank smaller or equal to 2 in $M_{k_1}\otimes M_{k_2...k_n}$

Now, if $A$ has  tensor rank 1 in $M_{k_1}\otimes M_{k_2\ldots k_n}$, let $A=C_1\otimes E_1$ such that $C_1\in M_{k_1}$ and $E_1\in M_{k_2...k_n}$ are positive semidefinite Hermitian matrices. By theorem \ref{theoremprincipal}, item a), $E_1$ is a linear combination of of $A_2\otimes\ldots\otimes A_n$ and $B_2\otimes \ldots\otimes B_n$. Notice that $E_1$ has tensor rank smaller or equal to 2 in $M_{k_2}\otimes\ldots\otimes M_{k_n}$, then by induction on $n$ the result follows.

Now, if $A$ in $M_{k_1}\otimes M_{k_2\ldots k_n}$ has  tensor rank 2, let $C_1\otimes E_1+D_1\otimes E_2$ be  a minimal separable decomposition of $A$ in $M_{k_1}\otimes M_{k_2\ldots k_n}$, by theorem \ref{theoreminimalseparabilitytensorrank2}. Thus, $C_1$ and $D_1$ are linear independent positive semidefinite Hermitian matrices.

By theorem \ref{theoremprincipal}, item a), $E_1$ and $E_2$ are linear combinations of $A_2\otimes\ldots\otimes A_n$ and $B_2\otimes \ldots\otimes B_n$. Thus, $E_1$ and $E_2$
are positive semidefinite hermitan matrices with tensor rank at most 2 in $M_{k_2}\otimes\ldots\otimes M_{k_n}$. Therefore by induction on $n$ the result follows.

\end{proof}

\section {Sharp Inequalities for Separability}

In this section we provide sharp inequalities that ensure separability  for SPC/PPTmatrices. We also proved that in some sense these inequalities are sharp.

\subsection{An Inequality for separability of SPC matrices.}
\indent\\

Since we have reduced the separability problem for SPC matrices to the set of weak irreducibe matrices (proposition \ref{reductionfinal}), we only need to provide a sharp inequality for SPC matrices which are weak irreducible. The idea is to use this inequality in each term of the split decomposition in order to obtain an inequality for an arbitrary SPC matrix.

Recall that we know the format of the Hermitian Schmidt decomposition for a weak irreducible SPC matrices (theorem \ref{conditionsirreducible2}). 

\begin{lemma}\label{lemmahalf} If $\gamma_1,\gamma_i\in M_k$ are orthonormal Hermitian matrices such that $\gamma_1$ is positive semidefinite, $\Im(\gamma_i)\subset\Im(\gamma_1)$ and $\mu$ is the minimal positive eigenvalue of $\gamma_1\otimes \gamma_1$ then $\frac{1}{2\mu}(\gamma_1\otimes \gamma_1)+\gamma_i\otimes\gamma_i$ is separable.
\end{lemma}
\begin{proof}
By theorem \ref{theoreminimalseparabilitytensorrank2}, we only need to prove that $\frac{1}{2\mu}(\gamma_1\otimes \gamma_1)+\gamma_i\otimes\gamma_i$ is positive semidefinite. Notice that the minimal eigenvalue of $\frac{1}{2\mu}(\gamma_1\otimes \gamma_1)$ is $\frac{1}{2}$. Since $\Im(\gamma_i)\subset\Im(\gamma_1)$ then $\Im(\gamma_i\otimes\gamma_i)\subset\Im(\gamma_1\otimes\gamma_1)$. Thus, to prove that $\frac{1}{2\mu}(\gamma_1\otimes \gamma_1)+\gamma_i\otimes\gamma_i$ is positive semidefinite is sufficient to prove that  the minimal eigenvalue of $\gamma_i\otimes\gamma_i$ is greater or equal to $-\frac{1}{2}$.

If $\gamma_i$ is positive semidefinite or negative semidefinite then $\gamma_i\otimes\gamma_i$ is positive semidefinite and the minimal eigenvalue is greater or equal to 0. Thus, suppose that $\gamma_i$ has positive and negative eigenvalues.

Let $a_1,\ldots, a_k$ be the eigenvalues of $\gamma_i$. Since $\gamma_i$ is normalized $a_1^2+\ldots+a_k^2=1$.  We know that these are real numbers and we know that the minimal eigenvalue of $\gamma_i\otimes\gamma_i$ is the product of the maximal eigenvalue (which is positive) by the minimal (which is negative). Suppose it is $a_1a_2$. 

Thus, we want to minimize the quadratic form $f:\mathbb{R}^k\rightarrow\mathbb{R}$, $f(a_1,\ldots,a_k)=a_1a_2$, subject to the restriction $a_1^2+\ldots+a_k^2=1$. We know that this minimal value is the minimal eigenvalue of the real symmetric matrix associated to the quadratic form which is $-\frac{1}{2}$.

\end{proof}

\begin{theorem}\label{inequality1} $($\textbf{The Inequality for SPC matrices}$)$ Let $A\in M_{k}\otimes M_k$ be a weak irreducible SPC matrix. Let $\sum_{i=1}^n\lambda_i\gamma_i\otimes\gamma_i$ be a Hermitian Schmidt decomposition of A with $\lambda_1\geq\ldots\geq\lambda_n>0$.
Let $\mu$ be the least positive eigenvalue of $\gamma_1\otimes\gamma_1$. If
$$\dfrac{\lambda_1\mu}{\lambda_2+...+\lambda_n}\geq \frac{1}{2}$$ 
then $A$ is separable.

\end{theorem}
\begin{proof}
Since $A$ is weak irreducible, by theorem \ref{conditionsirreducible2}, $\lambda_1>\lambda_i>0$ for $i\neq 1$ and $\Im(\gamma_i)\subset\Im(\gamma_1)$ for every $i$. Then $\Im(\gamma_i\otimes\gamma_i)\subset\Im(\gamma_1\otimes\gamma_1)$.

By lemma \ref{lemmapositive}, $\gamma_1\otimes\gamma_1$ is positive semidefinite.
Now if $\lambda_1\mu\geq\frac{1}{2}(\lambda_2+...+\lambda_n)$ then 
$A=
(\lambda_1 \mu-\sum_{i=2}^n\frac{\lambda_i}{2})(\frac{1}{\mu}\gamma_1\otimes\gamma_1)+\sum_{i=2}^n\lambda_i(\frac{1}{2\mu}\gamma_1\otimes\gamma_1+\gamma_i\otimes\gamma_i).$

All the matrices inside parentheses are separable by lemma \ref{lemmahalf}.

\end{proof}

\subsection{An Inequality for separability of PPT matrices.}
\indent\\

Again, we only need to provide a sharp inequality for weak irreducible PPT matrices. Then we can use this inequality in each term of the split decomposition in order to obtain an inequality for an arbitrary PPT matrix. 

Recall that a necessary condition for separability  of any matrix is to be PPT, then we actually obtained a sharp inequality for separability. 

Again, recall that we know the format of the Hermitian Schmidt decomposition of the weak irreducible PPT matrices (theorem \ref{conditionsirreducible2}).

\begin{theorem}\label{inequality2} $($\textbf{The Inequality for PPT matrices}$)$ Let $A\in M_{k}\otimes M_m$ be a Weak Irreducible PPT matrix. Let $\sum_{i=1}^n\lambda_i\gamma_i\otimes\delta_i$ be a Hermitian Schmidt decomposition of A with $\lambda_1\geq\ldots\geq\lambda_n>0$.
Let $\mu$ be the least positive eigenvalue of $\gamma_1\otimes\delta_1$. If
$$\dfrac{\lambda_1\mu}{\lambda_2+...+\lambda_n}\geq 1$$ 
then $A$ is separable.
\end{theorem}
\begin{proof}
Since $A$ is weak irreducible, by theorem \ref{conditionsirreducible2}, $\lambda_1>\lambda_i>0$ for $i\neq 1$ and $\Im(\gamma_i)\subset\Im(\gamma_1)$ and $\Im(\delta_i)\subset\Im(\delta_1)$ for every $i$. Then $\Im(\gamma_i\otimes\delta_i)\subset\Im(\gamma_1\otimes\delta_1)$.

 By lemma \ref{lemmapositive}, $\gamma_1\otimes\delta_1$ is positive semidefinite. Notice that the least positive eigenvalue of  $\frac{1}{\mu}\gamma_1\otimes\delta_1$ is $1$ and since $tr(\gamma_i^2)=tr(\delta_i^2)=1$ the least eigenvalue of $\gamma_i\otimes\delta_i$ is greater or equal to -1.  Thus, $\frac{1}{\mu}\gamma_1\otimes\delta_1+\gamma_i\otimes\delta_i$ is positive semidefinite and by theorem \ref{theoreminimalseparabilitytensorrank2} is separable.
Now if $\lambda_1\mu\geq \lambda_2+...+\lambda_n$ then 
$A=
(\lambda_1 \mu-\sum_{i=2}^n\lambda_i)(\frac{1}{\mu}\gamma_1\otimes\delta_1)+\sum_{i=2}^n\lambda_i(\frac{1}{\mu}\gamma_1\otimes\delta_1+\gamma_i\otimes\delta_i).$

Notice that all the matrices inside parentheses are separable.

\end{proof}

\subsection{The Inequalities are Sharp.}
\indent\\

Let us prove that the first inequality, in theorem \ref{inequality1} is sharp, comparing it to an inequality proved in \cite{leinaas}. Let $\gamma_1, \gamma_2, \gamma_3, \gamma_4$ be the orthonormal Pauli's basis of $M_2$, where $\gamma_1=\frac{1}{\sqrt{2}} Id$.

The authors of \cite{leinaas} showed that a matrix
$$ \gamma_1\otimes\gamma_1+d_2\gamma_2\otimes\gamma_2+d_3\gamma_3\otimes\gamma_3+d_4\gamma_4\otimes\gamma_4$$
is separable if and only if $|d_2|+|d_3|+|d_4|\leq 1.$

Thus, a matrix $A=\lambda_1\gamma_1\otimes\gamma_1+\lambda_2\gamma_2\otimes\gamma_2+\lambda_3\gamma_3\otimes\gamma_3+\lambda_4\gamma_4\otimes\gamma_4$ such that $\lambda_i>0$, for every $i$, is separable  if and only if $|\frac{\lambda_2}{\lambda_1}|+|\frac{\lambda_3}{\lambda_1}|+|\frac{\lambda_4}{\lambda_1}|\leq 1.$ This is equivalent to $\dfrac{\lambda_1\mu}{\lambda_2+\lambda_3+\lambda_4}\geq \frac{1}{2}$, where $\mu=\frac{1}{2}$ is the least positive eigenvalue of $\gamma_1\otimes\gamma_1$.

Thus, we obtained the same inequality which is a necessary and sufficient condition for separability for this family of matrices.\\

Now let us prove that we can not replace $1$ in the inequality of theorem \ref{inequality2},  $\dfrac{\lambda_1\mu}{\lambda_2+...+\lambda_n}\geq 1$, by any $l$ such that $0<l<1$ and still obtain a sufficient condition for separability. In this sense this inequality is sharp.\\\\
Consider the matrix $A(n)\in M_{n+1\times n+1}\otimes M_{n+1\times n+1}$ defined by\\ $$A(n)=\lambda_1\gamma_1\otimes\gamma_1+\lambda_2\gamma_2\otimes(-\gamma_2)$$
such that $\lambda_1,\lambda_2>0$ and\\\\
$\gamma_1= 
$$\left[\begin{array}{ccllrr}
\frac{1}{\sqrt{n+1}}&0&\ldots&0\\
 0  & \frac{1}{\sqrt{n+1}} &\ldots&0\\
 \vdots & \vdots &\ddots&\vdots\\
 0&0&\ldots&\frac{1}{\sqrt{n+1}}\\
 \end{array}\right]$$
$,
$\gamma_2=
$$ \left[\begin{array}{ccllrr}
\frac{n}{\sqrt{n^2+n}}&0&\ldots&0\\
 0  & \frac{-1}{\sqrt{n^2+n}} &\ldots&0\\
 \vdots & \vdots &\ddots&\vdots\\
 0&0&\ldots&\frac{-1}{\sqrt{n^2+n}}\\
 \end{array}\right].$$\\\\
$

Notice that $tr(\gamma_1\gamma_2)=0$ and $tr(\gamma_1\gamma_1)=tr(\gamma_2\gamma_2)=1$. 
 
Since $A(n)$ has tensor rank two, a necessary and sufficient condition for its separability is positive semidefiniteness $($See theorem \ref{theoreminimalseparabilitytensorrank2}$)$, but this is equivalent to $\dfrac{\lambda_1}{\lambda_2}\dfrac{1}{n+1}\geq\dfrac{n^2}{n^2+n}$. 

The least eigenvalue of $\gamma_1\otimes \gamma_1$ is $\mu=\dfrac{1}{n+1}$. Therefore a necessary and sufficient condition for separability of $A(n)$ is $\dfrac{\lambda_1 \mu}{\lambda_2}\geq\dfrac{n^2}{n^2+n}$.
Thus for every $l$ $(0<l<1)$ we obtain a matrix $A(n)$, such that, a necessary and sufficient condition for its separability is $\dfrac{\lambda_1 \mu}{\lambda_2}\geq\dfrac{n^2}{n^2+n}>l$. 

Thus, we can not replace 1, in the inequality $\dfrac{\lambda_1 \mu}{\lambda_2+...+\lambda_n}\geq 1$, by any $l$ satisfing $0<l<1$ and still obtain a sufficient condition for separability.\\

\textbf{Acknowledgement}
\noindent  D. Cariello was supported by CNPq-Brazil Grant 245277/2012-9.

\medskip

\noindent\lbrack Daniel Cariello] Faculdade de
Matem\'{a}tica, Universidade Federal de Uberl\^{a}ndia, 38.400-902,
Uberl\^{a}ndia, Brasil\\
Facultad de Ciencias Matem\'{a}ticas, Plaza de Ciencias 3, Universidad Complutense de Madrid, Madrid, 28040, Spain.\\
e-mails: dcariello@famat.ufu.br, dcariell@ucm.es.

\end{document}